\documentclass{sig-alternate}
\usepackage{times}
\usepackage{amsfonts}
\usepackage{wrapfig}
\usepackage{lscape}
\usepackage{verbatim}
\usepackage{url}
\usepackage{color}
\usepackage{amsmath}
\usepackage{graphicx}
\usepackage{multirow}
\usepackage{booktabs}
\usepackage{listings}
\usepackage{ifthen}

\usepackage{caption}
\usepackage{subfig}

\usepackage{xcolor}
\usepackage{hyperref}

\hypersetup{
  pdftitle={Practical Massively Parallel Sorting},
  pdfauthor={Michael Axtmann, Timo Bingmann, Peter Sanders, and Christian Schulz},
  pdfsubject={distributed sorting, massively parallel, exascale sorting},
  colorlinks=true,
  pdfborder={0 0 0},
  bookmarksopen=true,
  bookmarksopenlevel=1,
  bookmarksnumbered=true,
  linkcolor=blue!60!black,
  citecolor=blue!60!black,
  urlcolor=blue!60!black,
  filecolor=green!60!black,
  pdfpagemode=UseNone,
  unicode=true,
}

\newcommand{\seq}[1]{\left\langle #1\right\rangle}

\newcommand{\ceil}[1]{\left\lceil #1\right\rceil}
\newcommand{\floor}[1]{\left\lfloor #1\right\rfloor}

\newcommand{\set}[1]{\left\{ #1\right\}}

\newcommand{\realrange}[2]{\left[#1, #2\right]}

\newcommand{\unitrange}[2]{\realrange{0}{1}}

\newcommand{\prob}[1]{{\mathbf{P}}\left[#1\right]}

\newcommand{\expect}{{\mathbf{E}}}

\newcommand{\Oh}[1]{\mathcal{O}\!\left( #1\right)}

\newcommand{\Th}[1]{\Theta\!\left( #1\right)}
\newcommand{\Om}[1]{\Omega\left(#1\right)}

\newcommand{\llabel}[1]{\label{\labelprefix:#1}}
\newcommand{\labelprefix}{} 

\newcommand{\discussionsize}{\small}

\newcommand{\frage}[1]{{[\bf #1]}\marginpar{$\bigotimes$}}

\newcommand{\punkt}{\enspace .}

\newenvironment{code}{\noindent\normalsize
\begin{tabbing}%
\hspace{2em}\=\hspace{2em}\=\hspace{2em}\=\hspace{2em}\=\hspace{2em}\=%
\hspace{2em}\=\hspace{2em}\=\hspace{2em}\=\hspace{2em}\=\hspace{2em}\=%
\kill}{\end{tabbing}}

\newcommand{\labelcommand}{}
\newcommand{\captiontext}{}
\newsavebox{\codeparam}
\newcounter{lineNumber}
\newenvironment{disscodepos}[3]{%
\renewcommand{\labelcommand}{#2}%
\renewcommand{\captiontext}{#3}%
\sbox{\codeparam}{\parbox{\textwidth}{#3}}%
\begin{figure}[#1]\begin{center}\begin{code}\setcounter{lineNumber}{1}}{%
\end{code}\end{center}\caption{\llabel{\labelcommand}\captiontext}\end{figure}}

{\end{disscodepos}}

\newcommand{\Procedure}{{\bf Procedure\ }}

\newcommand{\Dopar}       {{\bf dopar\ }}
\newcommand{\For}      {{\bf for\ }}

\newcommand{\Is}{\mbox{\rm := }}

\newcommand{\To}       {{\bf to\ }}

\newcommand{\If}       {{\bf if\ }}

\newcommand{\Then}     {{\bf then\ }}
\newcommand{\Else}     {{\bf else\ }}

\newcommand{\Return}   {{\bf return\ }}

\newcommand{\Rem}[1]   {{//\hspace{0.5mm}{\em{#1}}}}
\newcommand{\RRem}[1]   {\`{\bf //\hspace{0.5mm}~}{\em{#1}}}

\newdimen\endofsize\endofsize=0.5em

\newcommand{\ForFromToParallel}[3]{{\For $#1$ \Is $#2$ \To $#3$ \Dopar}}

\renewcommand{\Oh}[1]{\ensuremath{\mathcal{O}(#1)}}
\newcommand{\OhL}[1]{\ensuremath{\mathcal{O}\!\left(#1\right)}}

\renewcommand{\Om}[1]{\Omega(#1)}

\renewcommand{\Th}[1]{\Theta(#1)}

\newcommand{\Exch}{\mathrm{Exch}}
\newcommand{\ExchTilde}{\ensuremath{\smash{\mathrm{Exch}\llap{\raisebox{-3pt}{$\widetilde{\phantom{iih}}$}\hspace{3.8pt}}}}}
\newcommand{\Tstart}{\alpha}
\newcommand{\Tword}{\beta}
\newcommand{\Lmax}{L_{\max}}
\newcommand{\Lmin}{L_{\min}}

\newcommand{\postponed}[1]{}

\usepackage{tikz}
\usetikzlibrary{calc,matrix,positioning,patterns}

\usepackage{array}
\usepackage{booktabs}
\usepackage{placeins}

\definecolor{gray}{rgb}{0.5, 0.5, 0.5}
\newtheorem{theorem}{Theorem}
\newtheorem{lemma}{Lemma}

\renewcommand{\frage}[1]{}

\newcommand{\TODO}[1]{%
  {{\sf\textbf{TODO}}[\textit{#1}]}%
  \marginpar{%
    \raggedright\sf\textbf{TODO}\par%
  }}
\renewcommand{\TODO}[1]{}
\begin{document}
\pagestyle{plain}
\conferenceinfo{SPAA}{'15 Portland, Oregon USA}

\title{Practical Massively Parallel Sorting}

\numberofauthors{4}
\author{
\alignauthor
Michael Axtmann\\
       \affaddr{Karlsruhe Inst. of Technology}\\
       \affaddr{Karlsruhe, Germany}\\
       \email{michael.axtmann@kit.edu}
\alignauthor
Timo Bingmann\\
       \affaddr{Karlsruhe Inst. of Technology}\\
       \affaddr{Karlsruhe, Germany}\\
       \email{bingmann@kit.edu}
\alignauthor
Peter Sanders\\
       \affaddr{Karlsruhe Inst. of Technology}\\
       \affaddr{Karlsruhe, Germany}\\
       \email{sanders@kit.edu}
\and
\alignauthor 
Christian Schulz\\
       \affaddr{Karlsruhe Inst. of Technology}\\
       \affaddr{Karlsruhe, Germany}\\
       \email{christian.schulz@kit.edu}
}

\maketitle
\begin{abstract}
  Previous parallel sorting algorithms do not scale to the largest available
  machines, since they either have prohibitive communication volume or
  prohibitive critical path length.  We describe algorithms that are a viable
  compromise and overcome this gap both in theory and practice.  The algorithms
  are multi-level generalizations of the known algorithms sample sort and
  multiway mergesort. In particular our sample sort variant turns out to be very
  scalable.  Some tools we develop may be of independent interest -- a simple,
  practical, and flexible sorting algorithm for small inputs working in
  logarithmic time, a near linear time optimal algorithm for solving a
  constrained bin packing problem, and an algorithm for data delivery, that
  guarantees a small number of message startups on each processor.
\end{abstract}
\pagestyle{plain}
\thispagestyle{plain}




\section{Introduction}
\label{s:intro}

Sorting is one of the most fundamental non-numeric algorithms which is needed in
a multitude of applications. For example, load balancing in supercomputers often
uses space-filling curves. This boils down to sorting data by their position on
the curve for load balancing. Note that in this case most of the work is done for the
application and the inputs are relatively small. For these cases, we need
sorting algorithms that are not only asymptotically efficient for huge inputs
but as fast as possible down to the range where near linear speedup is out of the
question.

We study the problem of sorting $n$ elements evenly distributed over $p$
processing elements (PEs) numbered $1..p$.%
\footnote{We use the notation $a..b$ as a shorthand for $\set{a,\ldots,b}$.}
The output requirement is that the
PEs store a permutation of the input elements such that the elements on each PE
are sorted and such that no element on PE $i$ is larger than any elements on PE
$i+1$.

There is a gap between the theory and practice of parallel sorting
algorithms.  Between the 1960s and the early 1990s there has been intensive work
on achieving asymptotically fast and efficient parallel sorting algorithms. The
``best'' of these algorithms, e.g., Cole's celebrated $\Oh{\log p}$ algorithm
\cite{Col88}, have prohibitively large constant factors. Some simpler algorithms
with running time $\Oh{\log^2 p}$, however, contain interesting techniques that
are in principle practical. These include parallelizations of well known
sequential algorithms like mergesort and quicksort \cite{Jaj92}. However, when
scaling these algorithms to the largest machines, these algorithms cannot be
directly used since all data elements are moved a logarithmic number of times
which is prohibitive except for very small inputs.

For sorting large inputs, there are algorithms which have to move the data only
once. Parallel sample sort \cite{BleEtAl91short}, is a generalization of quicksort to
$p-1$ splitters (or pivots) which are chosen based on a sufficiently large sample of the
input.  Each PE partitions its local data into $p$ pieces using the splitters and
sends piece $i$ to PE $i$. After the resulting all-to-all exchange, each PE
sorts its received pieces locally. Since every PE at least has to receive the
$p-1$ splitters, sample sort can only by efficient for $n=\Om{p^2/\log p}$,
i.e., it has isoefficiency function $\Om{p^2/\log p}$ (see also \cite{KumEtAl94}).
Indeed, the involved constant factors may be fairly large since the all-to-all exchange
implies $p-1$ message startups if data exchange is done directly.

In parallel $p$-way multiway mergesort \cite{VSIR91,SSP07}, each PE first sorts
its local data. Then, as in sample sort, the data is partitioned into $p$ pieces
on each PE which are exchanged using an all-to-all exchange. Since the local
data is sorted, it becomes feasible to partition the data \emph{perfectly} so
that every PE gets the same amount of data.%
\footnote{Of course this is only possible up to rounding $n/p$ up or down. To
  simplify the notation and discussion we will often neglect these issues if
  they are easy to fix.}
Each PE receives $p$ pieces which have to be merged together.
Multiway mergesort has an even worse isoefficiency function due to the overhead for
partitioning.

Compromises between these two extremes -- high asymptotic scalability but
logarithmically many communication operations versus low scalability but only a single
communication -- have been considered in the BSP model
\cite{Val94}. Gerbessiotis and Valiant \cite{GerVal94} develop a multi-level BSP
variant of sample sort.  Goodrich \cite{Goo99} gives communication efficient
sorting algorithms in the BSP model based on multiway merging. However, these
algorithms needs a significant constant factor more communications per element
than our algorithms. Moreover, the BSP model allows arbitrarily fine-grained
communication at no additional cost. In particular, an implementation of the
global data exchange primitive of BSP that delivers messages directly has a
bottleneck of $p$ message startups for every global message exchange.  Also see
Section~\ref{ss:deliver} for a discussion why it is not trivial to adapt the BSP
algorithms to a more realistic model of computation -- it turns out that for
worst case inputs, one PE may have to receive a large number of small messages.

In Section~\ref{s:building} we give building blocks that may also be of
independent interest.  This includes a distributed memory algorithm for
partitioning $p$ sorted sequences into $r$ pieces each such that the
corresponding pieces of each sequence can be multiway merged independently.  We
also give a simple and fast sorting algorithm for very small inputs. This
algorithm is very useful when speed is more important than efficiency, e.g., for
sorting samples in sample sort. Finally, we present an algorithm for distributing
data destined for $r$ groups of PEs in such a way that all PEs in a group get the same amount of \emph{data} \emph{and} a similar amount of \emph{messages}.

Sections \ref{s:merge} and \ref{s:sample} develop multi-level
variants of multiway mergesort and sample sort respectively.  The basic tuning
parameter of these two algorithms is the number of (recursion) levels.  With $k$ levels, we
basically trade moving the data $k$ times for reducing the startup overheads to
$\Oh{k\sqrt[k]{p}}$.  Recurse last multiway mergesort (RLM-sort) described in
Section~\ref{s:merge} has the advantage of achieving perfect load balance.  The
adaptive multi-level sample sort (AMS-sort) introduced in Section~\ref{s:sample}
accepts a slight imbalance in the output but is up to a factor $\log^2 p$ faster
for small inputs.  A feature of AMS-sort that is also interesting for
single-level algorithms is that it uses overpartitioning.  This reduces the
dependence of the required sample size for achieving imbalance $\varepsilon$ from
$\Oh{1/\varepsilon^2}$ to $\Oh{1/\varepsilon}$. We have already outlined these
algorithms in a preprint~\cite{ABSS14}.

In Section~\ref{s:experiments} we report results of an experimental
implementation of both algorithms.  In particular AMS-sort scales up to
$2^{15}$ cores even for moderate input sizes. Multiple levels have a
clear advantage over the single-level variants.\frage{todo: more highlights?}

\section{Preliminaries}
\label{s:prelim}

For simplicity, we will assume all elements to have unique keys.  This is
without loss of generality in the sense that we enforce this assumption by an
appropriate tie breaking scheme. For example, replace a key $x$ with a triple
$(x,y,z)$ where $y$ is the PE number where this element is input and $z$ the
position in the input array. With some care, this can be implemented in such a
way that $y$ and $z$ do not have to be stored or communicated explicitly. In
Appendix~\ref{app:tie} we outline how this can be implemented efficiently for
AMS-sort.

\subsection{Model of Computation}\label{ss:model}

A successful realistic model is (symmetric) single-ported message passing: Sending
a message of size $\ell$ machine words takes time $\Tstart+\ell\Tword$. The parameter
$\Tstart$ models startup overhead and $\Tword$ the time to communicate a machine word. For
simplicity of exposition, we equate the machine word size with the size of a
data element to be sorted.  We use this model whenever possible. In particular,
it yields good and realistic bounds for collective communication operations. 
For example, we get time $\Oh{\Tword \ell+\Tstart\log p}$ for broadcast, reduction, and prefix sums
over vectors of length $\ell$ \cite{BalEtAl95,SST09}.
However, for moving the bulk of the
data, we get very complex communication patterns where it is difficult to
enforce the single-ported requirement.

Our algorithms are bulk synchronous. Such algorithms are often
described in the framework of the BSP model \cite{Val94}. However, it is not clear how to
implement the data exchange step of BSP efficiently on a realistic parallel
machine.  In particular, actual implementations of the BSP model deliver the
messages directly using up to $p$ startups. For massively parallel machines this
is not scalable enough.  The BSP$^*$ model \cite{BDH95} takes this into account
by imposing a minimal message size. However, it also charges the cost for the
maximal message size occurring in a data exchange for all its messages and this would be too expensive for
our sorting algorithms. We therefore use our own model:  We consider a black box
data exchange function $\Exch(P,h,r)$ telling us how long it takes to exchange
data on a compact subnetwork of $P$ PEs in such a way that no PE receives or
sends more than $h$ words in total and at most $r$ messages in total.  Note
that all three parameters of the function $\Exch(P,h,r)$ may be essential, as
they model locality of communication, bottleneck communication volume (see also
\cite{Borkar13,SSM13}) and startups respectively.
Sometimes we also write $\ExchTilde(P,h,r)$ as a shorthand for $(1+o(1))\Exch(P,h,r)$
in order to summarize a sum of $\Exch(\cdot)$ terms by the dominant one.
We will also use that to absorb terms of the form $\Oh{\Tstart\log p}$ and $\Oh{\Tword r}$
since these are obvious lower bounds for the data exchange as well.
A lower bound in the single-ported model
for $\Exch(P,h,r)$ is $h\Tword+r\Tstart$ if data is delivered directly. There are reasons to believe that we can come
close to this but we are not aware of actual matching upper
bounds. There are offline scheduling algorithms which can deliver the data using
time $h\Tword$ when startup overheads are ignored (using edge coloring of
bipartite multi-graphs). However, this chops messages into many blocks and also
requires us to run a parallel edge-coloring algorithm. 

\subsection{Multiway Merging and Partitioning}

Sequential multiway merging of $r$ sequences with total length $N$ can be done
in time $\Oh{N\log r}$. An efficient practical implementation may use tournament
trees \cite{Knu98short,San00b,SSP07}. If $r$ is small enough, this is even cache
efficient, i.e., it incurs only $\Oh{N/B}$ cache faults where $B$ is the cache
block size. If $r$ is too large, i.e., $r>M/B$ for cache size $M$, then a
multi-pass merging algorithm may be advantageous. One could even consider a cache
oblivious implementation \cite{BFV04}.

The dual operation for sample sort is partitioning the data according to $r-1$
splitters. This can be done with the same number of comparisons and similarly
cache efficiently as $r$-way merging but has the additional advantage that it
can be implemented without causing branch mispredictions \cite{SW04}.

\section{More Related Work}\label{s:related}

\begin{figure*}\normalsize\centering
  \input{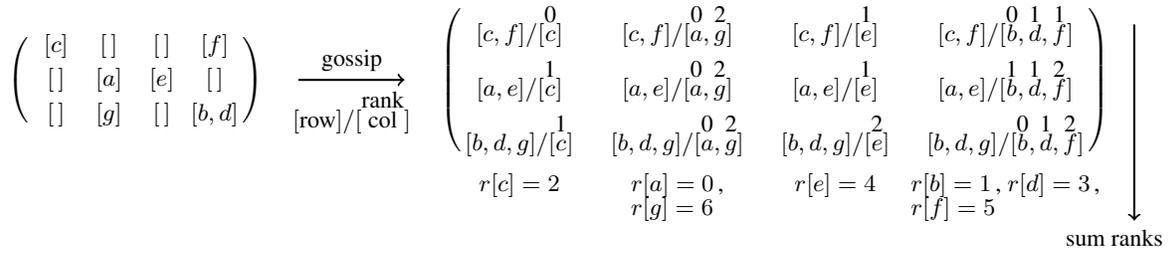}
  \caption{Example calculations done during fast work inefficient sorting algorithm on a $3 \times 4$ array of processors. The entries in the matrix on the right show elements received from the particular row and column during the allGather, and the corresponding calculated ranks.}\label{fig:fast inefficient}
\end{figure*}

\begin{figure}[b!]
\begin{code}
\Rem{select element with global rank $k$}\\
\Procedure multiSelect$(d_1,\ldots,d_{p}, k)$\+\\
  \If $\sum_{1\leq i\leq p}|d_i| = 1$ \Then \RRem{base case}\+\\
    \Return the only nonempty element\-\\
  select a pivot $v$\RRem{e.g. randomly}\\
  \ForFromToParallel{i}{1}{p}\+\\
    find $j_i$ such that $d_i[1..j_i] < v$ and $d[j_i+1..]\geq v$\-\\
  \If $\sum_{1\leq i\leq p}|j_i| \geq k$ \Then\+\\ \Return multiSelect$(d_1[1..j_1], \ldots, d_{p}[1..j_{p}], k)$\-\\
  \Else\+\\ \Return multiSelect$(d_1[j_1+1..], \ldots, d_{p}[j_{p}+1..],$\+\\
              $k-\sum_{0\leq i<p}|j_i|)$\\
\end{code}
\caption{\label{alg:musel}Multisequence selection algorithm.}
\end{figure}

Li and Sevcik \cite{LiSev94} describe the idea of overpartitioning.  However,
they use centralized sorting of the sample and a master worker load balancer
dealing out buckets for sorting in order of decreasing bucket size. 
This leads to very good load balance but is not scalable enough for our purposes and
heuristically disperses buckets over all PEs. Achieving the more strict output
format that our algorithm provide would require an additional complete data exchange.
Our AMS-sort from Section~\ref{s:sample} is fully parallelized without sequential 
bottlenecks and optimally partitions consecutive ranges of buckets.

A state of the art practical parallel sorting algorithm is described by
Solomonik and Kale \cite{SolKal10}. This single level algorithm can be viewed as
a hybrid between multiway mergesort and (deterministic) sample sort.
Sophisticated measures are taken for overlapping internal work and
communication.  TritonSort \cite{RasEtAl11} is a very successful sorting
algorithm from the database community. TritonSort is a version of single-level
sample-sort with centralized generation of splitters.

\section{Building Blocks}\label{s:building}

\subsection{Multisequence Selection}\label{ss:mss}

In its simplest form, given sorted sequences $d_1,\ldots,d_p$ and a rank $k$,
multisequence selection asks for finding an element $x$ with rank $k$ in the
union of these sequences. If all elements are different, $x$ also defines
positions in the sequences such that there is a total number of $k$ elements to
the left of these positions.

There are several algorithms for multisequence selection,
e.g. \cite{VSIR91,SSP07}.  Here we propose a particularly simple and intuitive
method based on an adaptation of the well-known quick-select algorithm
\cite{Hoa61b,MehSan08}.  This algorithm may be folklore. The algorithm has also
been on the publicly available slides of Sanders' lecture on parallel algorithms
since 2008 \cite{San08VLPA}.  Figure~\ref{alg:musel} gives high level pseudo
code.  The base case occurs if there is only a single element (and $k=1$).
Otherwise, a random element is selected as a pivot. This can be done in parallel
by choosing the same random number between 1 and $\sum_{i}|d_i|$ on all
PEs. Using a prefix sum over the sizes of the sequences, this element can be
located easily in time $\Oh{\Tstart\log p}$. Where ordinary quickselect has to
partition the input doing linear work, we can exploit the sortedness of the
sequences to obtain the same information in time $\Oh{\log D}$ with $D :=
\max_i|d_i|$ by doing binary search in parallel on each PE. If items are evenly
distributed, we have $D = \Th{\frac{n}{p}}$, and thus only time
$\mathcal{O}(\log\frac{n}{p})$ for the search, which partitions all the
sequences into two parts.  Deciding whether we have to continue searching in the
left or the right parts needs a global reduction operation taking time
$\Oh{\Tstart\log p}$. The expected depth of the recursion is
$\Oh{\log\sum_i|d_i|}=\Oh{\log n}$ as in ordinary quickselect. Thus, the overall
expected running time is $\Oh{(\Tstart\log p+\log\frac{n}{p})\log n}$.

In our application, we have to perform $r$ simultaneous executions of multisequence
selection
on the same input sequences but on $r$ different rank values. The involved
collective communication operations will then get a vector of length $r$ as
input and their running time using an asymptotically optimal implementation
is $\Oh{r\Tword+\Tstart\log p}$ \cite{BalEtAl95,SST09}. Hence, the overall
running time of multisequence selection becomes
\begin{equation}\label{eq:time multiselect}
  \Oh{(\Tstart\log p+r\Tword+r\log\tfrac{n}{p})\log n}\punkt
\end{equation}

\subsection{Fast Work Inefficient Sorting}\label{ss:fastsort}

We generalize an algorithm from \cite{IKS09} which may also be considered
folklore.  In its most simple form, the algorithm arranges $n^2$ PEs as a square
matrix using PE indices from $1..n\times 1..n$. Input element $i$ is assumed to
be present at PE $(i,i)$ initially. The elements are first broadcast along rows
and columns. Then, PE $(i,j)$ computes the result of comparing elements $i$ and
$j$ (0 or 1). Summing these comparison results over row $i$ yields the rank of
element $i$.

Our generalization works for a rectangular $a\times b$ array of processors
where $a=\Oh{\sqrt{p}}$ and $b=\Oh{\sqrt{p}}$. In particular, when $p=2^P$ is a
power of two, then $a=2^{\ceil{P/2}}$ and $b=2^{\floor{P/2}}$.  Initially, there
are $n$ elements uniformly distributed over the PEs, i.e. each PE has at most
$\ceil{n/p}$ elements as inputs. These are first sorted locally in time
$\Oh{\frac{n}{p}\log\frac{n}{p}}$.

Then the locally sorted elements are gossiped (allGather) along both rows and
columns (see Figure~\ref{fig:fast inefficient}),
making sure that the received elements are sorted. This can be
achieved in time \linebreak\mbox{$\Oh{\Tstart\log p+\Tword\frac{n}{\sqrt{p}}}$}. For example, if the
number of participating PEs is a power of two, we can use the well known
hypercube algorithm for gossiping (e.g., \cite{KumEtAl94}). The only
modification is that received sorted sequences are not simply concatenated but
merged.
\footnote{For general $p$, we can also use a gather algorithm along a binary
  tree and finally broadcast the result.}

\begin{figure*}\normalsize\centering
  \input{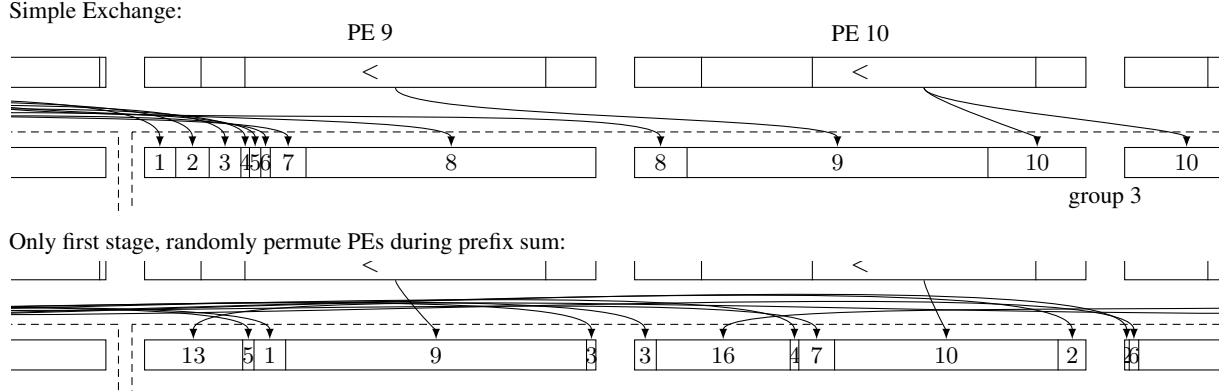}
  \caption{Exchange schema without and with first stage: permutation of PEs}\label{fig:exchange schema1}
\end{figure*}

\begin{figure*}[t]\normalsize\centering
  \input{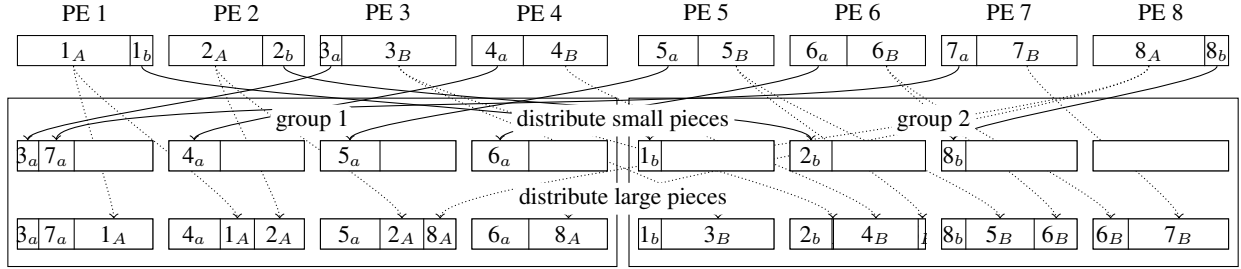}
  \vspace{1ex}
  \caption{Deterministic data delivery schema}\label{fig:deterministic exchange schema}
\end{figure*}

Elements received from column $i$ are then ranked with  respect to the elements received
from row $j$. This  can be done in time $\Oh{\frac{n}{\sqrt{p}}}$ by merging these
two sequences. Summing these local ranks along rows then yields the global rank of each
element.  If desired, this information can then be used for routing the input
elements in such a way that a globally sorted output is achieved. In our
application this is not necessary because we want to extract
elements with certain specified ranks as a sample\frage{todo: check. Answer: We just propagate those elements (to all PEs) whose rank fits.}.
Either way, we get overall execution time
\begin{equation}
\OhL{\Tstart\log p+\Tword\tfrac{n}{\sqrt{p}}+\tfrac{n}{p}\log\tfrac{n}{p}}\punkt
\end{equation}
Note that for $n$ polynomial in $p$ this bound is $\Oh{\Tstart\log
  p+\Tword\frac{n}{\sqrt{p}}}$. This restrictions is fulfilled for all
reasonable applications of this sorting algorithm.

~

\subsection{Delivering Data to the Right Place}\label{ss:deliver}

In the sorting algorithms considered here we face the following data
redistribution problem: Each PE has partitioned its locally present data into
$r$ pieces. The pieces with number $i$ have to be moved to PE group $i$ which
consists of PEs $(i-1)r+1..ir$. Each PE in a group should receive the same
amount of data except for rounding issues. 

We begin with a simple approach and then refine it in order to handle bad cases.
The basic idea is to compute a prefix sum over the
piece sizes -- this is a vector-valued prefix sum with vector length $r$.  As a
result, each piece is labeled with a range of positions within the group it
belongs to. Positions are numbers between $1$ and $m_i$
where $m_i\leq n/r$ is the number of elements assigned to group $i$. 
An element with number $j$ in group $i$ is sent to PE $(i-1)\frac{p}{r}+\lceil\frac{j}{m_i}\rceil$.
This way, each PE sends exactly $n/p$ elements and receives
at most $\ceil{m_ir/p}$ elements. Moreover, each piece is sent to one
or two target PEs responsible for handling it in the recursion.  Thereby, each
PE \emph{sends} at most $2r$ messages for the data exchange. Unfortunately, the
number of \emph{received messages}, although the same on the average, may vary
widely in the worst case. There are inputs where some PEs have to \emph{receive} $\Om{p}$
very small pieces. This happens when many consecutively numbered PEs send only
very small pieces of data (see PE 9 in the top of Figure~\ref{fig:exchange schema1}).

One way to limit the number of received pieces is to use randomization. We
describe how to do this while keeping the data perfectly balanced.  We describe
this approach in two stages where already the first, rather simple stage gives
a significant improvement.  The first stage is to choose the PE-numbering used
for the prefix sum as a (pseudo)random permutation within each group (see
Appendix~\ref{app:randperm}). However, it can be shown that if all but $p/r$
pieces are very small, this would still imply a logarithmic factor higher
startup overheads for the data exchange at some PEs. In
Appendix~\ref{app:randomDelivery} we give an advanced randomized algorithm that
gets rid of this logarithmic factor.  But now we give an algorithm that is
deterministic and at least conceptually simpler.

\subsubsection{A Deterministic Solution}\label{sss:deterministic}

The basic idea is to distribute small and large pieces separately. 
In Figure~\ref{fig:deterministic exchange schema} we illustrate the process.
First, small pieces of size at most $n/2pr$ are enumerated using a prefix
sum. Small piece $i$ of group $j$ is assigned to PE $\floor{i/r}$ of group
$j$. This way, all small pieces are assigned without having to split them and no
receiving PE gets more than half its final load.

In the second phase, the remaining (large) pieces are assigned taking the
residual capacity of the PEs into account. PE $i$ sends the description of its
piece for group $j$ to PE $\floor{i/r}$ of group $j$. This can be done in time
$\Exch(p,\Oh{r},r)$.  Now, each group produces an assignment of its large pieces
independently, i.e., each group of $p/r$ PEs assigns up to $p$ pieces
-- $r$ on each PE. In the following, we describe the assignment process for
a single group.

\begin{figure*}[t]\normalsize\centering
  \input{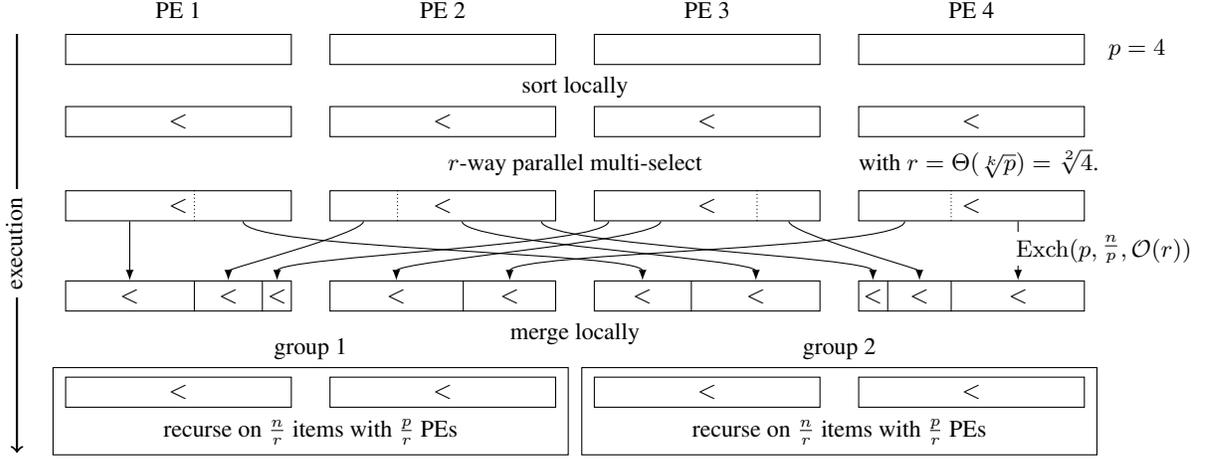}
  \vspace*{6pt}
  \caption{Algorithm schema of Recurse Last Parallel Multiway Mergesort}\label{fig:rlm-schema}
\end{figure*}

Conceptually, we enumerate the unassigned elements on the one hand and the
unassigned slots able to take them on the other hand and then map element $i$
to slot $i$.%
\footnote{A similar approach to data redistribution is described in
  \cite{HMS15}. However, here we can exploit special properties of the input to
  obtain a simpler solution that avoids segmented gather and scatter
  operations.} To implement this, we compute a prefix sum of the residual
capacities of the receiving PEs on the one hand and the sizes of the unassigned
pieces of the other hand. This yields two sorted sequences $X$ and $Y$
respectively which are merged in order to locate the destination PEs of each
large piece. Assume that ties in values of $X$ and $Y$ are broken such that
elements from $X$ are considered smaller.  In the merged sequence, a
subsequence of the form $\seq{x_i,y_j,\ldots,y_{j+k},x_{i+1},z}$ indicates that
pieces $j,\ldots,j+k$ have to moved to PE $i$. Piece $j+k$ may also wrap over
to PE $i+1$, and possibly to PE $i+2$ if $z=x_{i+2}$. The assumptions on the
input guarantee that no further wrapping over is possible since no piece can be
larger than $n/p$ and since every PE has residual capacity at least
$\frac{n}{2p}$.  Similarly, since large pieces have size at least $n/2pr$ and
each PE gets assigned at most $n/p$ elements, no PE gets more than
$\frac{n}{p}/\frac{n}{2pr}=2r$ large pieces.

The difficult part is merging the two sorted sequences $X$ and $Y$. Here one
can adapt and simplify the work efficient parallel merging algorithm for EREW
PRAMs from \cite{HagRue89}.  Essentially, one first merges the $p/r$ elements
of $X$ with a deterministic sample of $Y$ -- we include the prefix sum for the
first large piece on each PE into $Y$.  This merging operation can be done in
time $\Oh{\Tstart\log(p/r)}$ using Batcher's merging network \cite{Batcher68}.
Then each element of $X$ has to be located within the $\leq r$ local elements
of $Y$ on one particular PE.  Since it is impossible that these pieces (of
total size $\leq rn/p$) fill more than $2r$ PEs (of residual capacity $>n/2p$),
each PE will have to locate only $\Oh{r}$ elements. This can be done using
local merging in time $\Oh{r}$. In other words, the special properties of the
considered sequence make it unnecessary to perform the contention resolution
measures making \cite{HagRue89} somewhat complicated.  Overall, we get the
following deterministic result (recall from Section~\ref{ss:model} that
$\ExchTilde(\cdot)$ also absorbs terms of the form
$\Oh{\Tstart\log p+\Tword r}$).

\begin{theorem}
Data delivery of $r\times p$ pieces to $r$ parts can be implemented to run in time
$$\ExchTilde(p,\tfrac{n}{p},\Oh{r})\punkt$$
\end{theorem}

\section{Generalizing Multilevel Mergesort (RLM-Sort)}\label{s:merge}

We subdivide the PEs into ``natural'' groups of size $p'$ on which we want to
recurse.  Asymptotically, $r \Is p/p'$ around $\sqrt[k]{p}$ is a good
choice if we want to make $k$ levels of recursion. However, we may also fix $p'$
based on architectural properties. For example, in a cluster of many-core
machines, we might chose $p'$ as the number of cores in one node. Similarly, if
the network has a natural hierarchy, we will adapt $p'$ to that situation. For
example, if PEs within a rack are more tightly connected than inter-rack
connections, we may choose $p'$ to be the number of PEs within a rack. Other
networks, e.g., meshes or tori have less pronounced cutting points. However, it
still makes sense to map groups to subnetworks with nice properties, e.g.,
nearly cubic subnetworks. For simplicity, we will assume that $p$ is divisible
by $p'$, and that $r = \Th{\sqrt[k]{p}}$.

There are several ways to define multilevel multiway mergesort. We describe a
method we call ``recurse last'' (see Figure~\ref{fig:rlm-schema}) that needs to
communicate the data only $k$ times and avoids problems with many small
messages.  Every PE sorts locally first. Then each of these $p$ sorted
sequences is partitioned into $r$ pieces in such a way that the sum of these
piece sizes is $n/r$ for each of these $r$ resulting \emph{parts}. In contrast
to the single level algorithm, we run only $r$ multisequence selections in
parallel and thus reduce the bottleneck due to multisequence selection by a
factor of $p'$.

Now we have to move the data to the responsible groups.  We defer to
Section~\ref{ss:deliver} which shows how this is possible using time
$\ExchTilde(p,\frac{n}{p},\Oh{\sqrt[k]{p}})$.

Afterwards, group $i$ stores elements which are no larger than any element in
group $i+1$ and it suffices to recurse within each group. However, we do not
want to ignore the information already available -- each PE stores not an
entirely unsorted array but a number of sorted sequences. This information is
easy to use though -- we merge these sequences locally first and obtain locally
sorted data which can then be subjected to the next round of splitting.

\begin{figure*}[t]\normalsize\centering
  \input{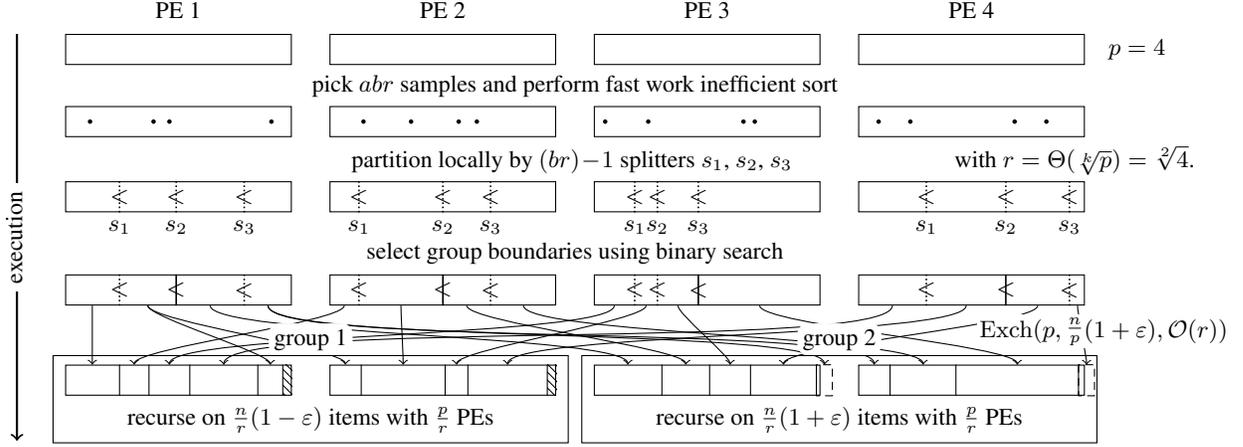}
  \vspace{1ex}
  \caption{Algorithm schema of AMS-sort}\label{fig:algo schema ams}
\end{figure*}

\begin{theorem}
  RLM-sort with $k=\Oh{1}$ levels of
  recursion can be implemented to run in time
  \begin{align}
  &\OhL{\left(\Tstart\log p+\sqrt[k]{p}\; \Tword + \sqrt[k]{p}\, \log\tfrac{n}{p}+\tfrac{n}{p}\right)\log n}+\nonumber\\
   & {\sum_{i=1}^k}\,\ExchTilde\left(p^{\frac{i}{k}},\tfrac{n}{p},\OhL{\sqrt[k]{p}}\right)\label{eq:RLM}\punkt
  \end{align}
\end{theorem}
\begin{proof}(Outline) Local sorting takes time $\Oh{\frac{n}{p}\log n}$.  For
  $k=\Oh{1}$ multiselections we get the bound from Equation~(\ref{eq:time
    multiselect}), $\Oh{(\Tstart\log p+r\Tword+r\log\tfrac{n}{p})\log n}$.
  Summing the latter two contributions, we get the first term of
  Equation~\eqref{eq:RLM}.

  In level $i$ of the recursion we have $r^i$ independent groups containing
  $\frac{p}{r^i} = \frac{p}{p^{i/k}} = p^{1 - \frac{i}{k}}$ PEs each. An
  exchange within the group in level $i$ costs
  $\ExchTilde(p^{1-i/k}, \frac{n}{p}, \Oh{\frac{r}{r^i}})$ time. Since all
  independent exchanges are performed simultaneously, we only need to sum over
  the $k$ recursive levels, which yields the second term of
  Equation~\eqref{eq:RLM}.
\end{proof}
Equation~(\ref{eq:RLM}) is a fairly complicated expression but using some
reasonable assumptions we can simplify it. If all communications are equally
expensive, the sum becomes $k\ExchTilde(p,\frac{n}{p},\Oh{\sqrt[k]{p}})$ -- we
have $k$ message exchanges involving all the data but we limit the number of
startups to $\Oh{\sqrt[k]{p}}$. On the other hand, on mesh or torus networks,
the first (global) exchange will dominate the cost and we get
$\ExchTilde(p,\frac{n}{p},\Oh{\sqrt[k]{p}})$ for the sum.  If we also assume
that data is delivered directly, $\Omega(\sqrt[k]{p})$ startups hidden in the
$\ExchTilde()$ term will dominate the $\Oh{\log^2p}$ startups in the remaining
algorithm.  We can assume that $n$ is bounded by a polynomial in $p$ --
otherwise, a traditional single-phase multi-way mergesort would be a better
algorithm. This implies that $\log n=\Th{\log p}$.  Furthermore, if
$n=\omega(p^{1+1/k}\log p)$ then $n/p=\omega(p^{\frac{1}{k}}\log p)$, and the
term $\Om{\Tword\frac{n}{p}}$ hidden in the data exchange term dominates the
term $\Oh{\Tword p^{\frac{1}{k}}\log n}$.  Thus Equation~(\ref{eq:RLM})
simplifies to $\Oh{\frac{n}{p}\log n}$ (essentially the time for internal
sorting) plus the data exchange term.

If we also assume $\Tstart$ and $\Tword$ to be constants and estimate
$\ExchTilde$-term as $\Oh{\frac{n}{p}}$, we get 
execution time 
$$\Oh{\sqrt[k]{p}\log^2p+\frac{n}{p}\log n}\punkt$$
From this, we can infer a $\Oh{p^{1+1/k}\log p}$ as isoefficiency function.

\section{Adaptive Multi-Level\\ Sample Sort (AMS-Sort)}\label{s:sample}

A good starting point is the multi-level sample sort algorithm by Gerbessiotis
and Valiant \cite{GerVal94}. However, they use centralized sorting of the
sample and their data redistribution may lead to some processors receiving
$\Om{p}$ messages (see also Section~\ref{ss:deliver}). We improve on this
algorithm in several ways to achieve a truly scalable algorithm. First, we sort
the sample using fast parallel sorting.  Second, we use the advanced data
delivery algorithms described in Section~\ref{ss:deliver}, and third, we give a
scalable parallel adaptation of the idea of overpartitioning \cite{LiSev94} in
order to reduce the sample size needed for good load balance.\TODO{Note Timo:
  also better cache footprint on the nodes.}

But back to our version of multi-level sample sort (see Figure~\ref{fig:algo
  schema ams}). As in RLM-sort, our intention is to split the PEs into $r$
groups of size $p'=p/r$ each, such that each group processes elements with
consecutive ranks.  To achieve this, we choose a random sample of size $abr$
where the \emph{oversampling factor} $a$ and the \emph{overpartitioning factor}
$b$ are tuning parameters.  The sample is sorted using a fast sorting
algorithm.  We assume the fast inefficient algorithm from
Section~\ref{ss:fastsort}. Its execution time is
$\Oh{\frac{abr}{p}\log\frac{abr}{p}+\Tword\frac{abr}{\sqrt{p}}+\Tstart\log p}$.

From the sorted sample, we choose $br-1$ splitter elements with equidistant
rank. These splitters are broadcast to all PEs. This is possible in time
$\Oh{\Tword br+\Tstart\log p}$.

Then every PE partitions its local data into $br$ \emph{buckets} corresponding
to these splitters. This takes time $\Oh{\frac{n}{p}\log(br)}$.

Using a global (all-)reduction, we then determine global bucket sizes in time
$\Oh{\Tword br+\Tstart\log p}$. These can be used to assign buckets to
PE-groups in a load balanced way: Given an upper bound $L$ on the number of
elements per PE-group, we can scan through the array of bucket sizes and skip
to the next PE-group when the total load would exceed $L$.  Using binary search
on $L$ this finds an optimal value for $L$ in time $\Oh{br\log n}$ using a
sequential algorithm.  In Appendix~\ref{app:scan} we explain how this can be
improved to $\Oh{br\log br}$ and, using parallelization, even to
$\Oh{br+\Tstart\log p}.$

\frage{or even better using parallization?}
\begin{lemma}
The above binary search scanning algorithm indeed finds the optimal $L$.
\end{lemma}
\begin{proof}
We first show that binary search suffices to find the optimal $L$ for which the
scanning algorithm succeeds. Let $L^*$ denote
this value. For this to be true, it suffices to show that for any $L\geq L^*$, the scanning
algorithm finds a feasible partition into groups with load at most $L$.
This works because  the scanning algorithm maintains the invariant that after defining $i$ groups,
the algorithm with bound $L$ has scanned at least as many buckets as the algorithm with bound $L^*$.
Hence, when the scanning algorithm with bound $L^*$ has defined all groups, the one with 
bound $L$ has scanned at least as many buckets as the algorithm with bound $L^*$.
Applying this invariant to the final group yields the desired result.

Now we prove that no other algorithm can find a better solution.
Let $L^*$ denote the maximum group size of an optimal partitioning algorithm.
We argue that the scanning algorithm with bound $L^*$ will succeed.
We now compare any optimal algorithm with the scanning algorithm. 
Consider the first $i$ buckets defined by both algorithms.
It follows by induction on $i$ that the total size $s^s_i$ of these buckets for the scanning algorithm
is at least as large as the corresponding value $s^*_i$ for the optimal algorithm:
This is certainly true for $i=0$ ($s^s_0=s^*_0=0$). 
For the induction step, suppose that the optimal algorithm chooses a bucket of size $y$, i.e., 
$s^*_{x+1}=s^*_x+y$. By the induction hypothesis, we know that $s^s_i\geq s^*_i$. Now suppose,
the induction invariant would be violated for $i+1$, i.e., $s^s_{i+1}<s^*_{i+1}$. 
Overall, we get $s^*_i\leq s^s_i < s^s_{i+1} < s^*_{i+1}$.
This implies that
$s^s_{i+1}-s^s_i$ -- the size of group $i+1$ for the scanning algorithm -- is smaller than $y$. 
Moreover, this group contains a proper subset of the buckets included by the optimal algorithm.
This is a impossible since there is no reason why the scanning algorithm should not at least 
achieve a bucket size $s^*_{i+1}-s^s_i\leq y\leq L^*$.
\end{proof}
\begin{lemma}
\label{lm:ams-param}
We can achieve $L=(1+\varepsilon)\frac{n}{r}$ with high probability
choosing appropriate $b=\Om{1/\varepsilon}$ and $ab=\Om{\log r}$.
\end{lemma}
\begin{proof}We only give the basic idea of a proof.
  We argue that the scanning algorithm is likely to succeed with
  $L=(1+\varepsilon)\frac{n}{r}$ as a group size limit.
  Using Chernoff bounds it can be shown that
  $ab=\Om{\log p}$ ensures that no bucket has size larger than $\frac{n}{r}$ with
  high probability.
  Hence, the scanning algorithm can always build feasible PE
  groups from one or multiple buckets.  

  Choosing $b\geq 2/\varepsilon$ means that
  the expected bucket size is $\leq \frac{\varepsilon}{2}\cdot\frac{n}{r}$.
  Indeed, most elements will be in buckets of size less than $\varepsilon\frac{n}{r}$.
  Hence, when the scanning algorithm adds a bucket to a PE-group such that 
  the average group size $\frac{n}{r}$ is passed for the first time,
  most of the time this additional group will also fit below the limit of
  $(1+\varepsilon)\frac{n}{r}$. Overall, the scanning algorithm will mostly form
  groups of size exceeding $\frac{n}{r}$ and thus $r$ groups will suffice to
  cover all buckets of total size $n$. 
\frage{more formal proof?}
\end{proof}

The data splitting defined by the bucket group is then the input for the
data delivery algorithm described in Section~\ref{ss:deliver}.
This takes time
$\ExchTilde\left(p,(1+o(1))L,(2+o(1))r\right)$).

We recurse on the PE-groups similar to Section~\ref{s:merge}. Within
the recursion it can be exploited that the elements are already partitioned into
$br$ buckets.

We get the following overall execution time for one level:
\begin{lemma}\label{lem:ams}
One level of AMS-sort works in time
  \begin{align}
\OhL{\frac{n}{p}\log\frac{r}{\varepsilon}+\Tword\frac{r}{\varepsilon}}+
\ExchTilde(p,(1+\varepsilon)\tfrac{n}{p},\Oh{r})
\label{eq:AMS}\punkt
  \end{align}
\end{lemma}
\begin{proof}(Outline)
This follows from Lemma~\ref{lm:ams-param} and the individual running times
described above using $ab=\Th{\max(\log r, 1/\varepsilon)}$,
$b=\Th{1/\varepsilon}$, and fast inefficient sorting for sorting the sample. The
sample sorting term then reads
$\Oh{\frac{abr}{p}\log\frac{abr}{p}+\Tword\frac{abr}{\sqrt{p}}+\Tstart\log p}$
which is $o(\frac{n}{p}\log\frac{r}{\epsilon}+\frac{\Tword}{\varepsilon})+\Tstart\log p$.
Note that the term $\Tstart\log p$ is absorbed into the \ExchTilde-term.
\end{proof}

Compared to previous implementations of sample sort, including the one from Gerbessiotis and Valiant \cite{GerVal94},
AMS-sort improves the sample size from $\Oh{p\log p/\varepsilon^2}$ to
$\Oh{p(\log r+1/\varepsilon)}$ and the number of startup overheads in the $\Exch$-term
from $\Oh{p}$ to $\Oh{r}$.

In the base case of AMS-sort, when the recursion reaches a single
PE, the local data is sorted sequentially.

\begin{theorem}\label{thm:ams}
Adaptive multi-level sample sort (AMS-sort) with $k$ levels of
  recursion and a factor $(1+\varepsilon)$ imbalance in the output
can be implemented to run in time
  \begin{align*}
  \OhL{\frac{n}{p}\log n + \Tword \frac{k^2\sqrt[k]{p}}{\varepsilon}}+
   {\sum_{i=1}^k}\,\ExchTilde\left(p^{\frac{i}{k}},(1+\varepsilon)\tfrac{n}{p},\OhL{\sqrt[k]{p}}\right)
  \end{align*}
if $k=\Oh{\log p/\log\log p}$ and $\frac{1}{\varepsilon}=\Oh{\sqrt[k]{n}}$.
\end{theorem}
\begin{proof}
We choose $r=\sqrt[k]{p}$.  Since errors multiply, we choose
$\varepsilon'=\sqrt[k]{1+\varepsilon}-1=\Th{\frac{\varepsilon}{k}}$ as the balance parameter for
each level. Using Lemma~\ref{lem:ams} we get the following terms.\\
For internal computation:
$\Oh{\frac{n}{p}}\log n$ for the final internal sorting.
(We do not exploit that overpartitioning presorts the data to some extent.)
For partitioning, we apply Lemma~\ref{lem:ams} and get time
\begin{align}
\OhL{k\log\frac{r}{\epsilon'}}&=\OhL{k\frac{n}{p}\log\frac{k\sqrt[k]{p}}{\varepsilon}}\nonumber\\
  &=\frac{n}{p}\OhL{\log p+k\log k +k\log\frac{1}{\varepsilon}}\label{eq:detailedAMS}\\
  &=\frac{n}{p}\OhL{\log p+\log n}\nonumber
\punkt
\end{align} 
The last estimate uses the preconditions $k=\Oh{\log p/\log\log p}$ and
$\frac{1}{\varepsilon}=\Oh{\sqrt[k]{n}}$ in order to simplify the theorem.

For communication volume we get
$k\cdot\Tword\frac{r}{\varepsilon'}=\Oh{\Tword \frac{k^2\sqrt[k]{p}}{\varepsilon}}$. For startup latencies we get 
$\Oh{\Tstart k\log p}$ which can be absorbed into the $\ExchTilde()$-terms.

The data exchange term is the same as for RLM-sort except that we have a slight
imbalance in the communication volume.
\end{proof}

Using a similar argument as for RLM-sort, for constant $k$ and $\varepsilon$, we get an isoefficiency function of
$p^{1+1/k}/\log p$ for $r=\sqrt[k]{p}$. This is a factor $\log^2 p$ better than for RLM-sort
and is an indication that AMS-sort might be the better algorithm -- in particular
if some imbalance in the output is acceptable and if the inputs are rather small.

Another indicator for the good scalability of AMS-sort is that we can view it as
a generalization of parallel quicksort that also works efficiently for very
small inputs.  For example, suppose $n=\Oh{p\log p}$ and $1/\varepsilon=\Oh{1}$.
We run $k=\Oh{\log p}$ levels of AMS-sort with $r=\Oh{1}$ and
$\varepsilon'=\Oh{k/\varepsilon}$.  This yields running time
$\Oh{\log^2p\log\log p+\Tstart\log^2p}$ using the bound from
Equation~(\ref{eq:detailedAMS}) for the local work.  This does a factor
$\Oh{\log\log p}$ more local work than an asymptotically optimal algorithm.
However, this is likely to be irrelevant in practice since it is likely that
$\Tstart\gg \log\log p$. Also the factor $\log\log p$
would disappear in an implementation that exploits the information gained during
bucket partitioning.

\section{Experimental Results}
\label{s:experiments}

We now present the results of our AMS-sort and RLM-sort experiments. In our
experiments we run a \emph{weak scaling} benchmark, which shows how the
wall-time varies for an increasing number of processors for a fixed amount of
elements per processor. 
Furthermore, in Appendix~\ref{app:furtherPlots} we show additional experiments
considering the effect of overpartitioning in more detail.
The test covers the AMS-sort and RLM-sort algorithms executed
with $10^5$, $10^6$, and $10^7$\,64-bit integers. We ran our experiments at the
thin node cluster of the SuperMUC
(\href{www.lrz.de/supermuc}{www.lrz.de/supermuc}), a island-based distributed
system consisting of $18$ islands, each with 512 computation nodes. 
However, the maximum number of islands available to us was four.
Each computation node has two Sandy Bridge-EP Intel Xeon E5-2680 8-core processors
with a nominal frequency of $2.7$\,GHz and $32$\,GByte of memory. However, jobs
will run at the standard frequency of $2.3$\,GHz as the LoadLeveler does not
classify the implementation as accelerative based on the algorithm's energy
consumption and runtime. A non-blocking topology tree connects the nodes within an island
using the Infiniband FDR10 network technology. Computation nodes are connected to
the non-blocking tree by Mellanox FDR ConnectX-3 InfiniBand mezzanine adapters. A pruned tree connects the islands
among each other with a bi-directional bi-section bandwidth ratio of $4:1$. The
interconnect has a theoretical bisection bandwidth of up to $35.6$
TB/s.

\subsection{Implementation Details}

We implemented AMS-sort and RLM sort in C++ with the main objective to
demonstrate that multilevel algorithms can be useful for large $p$ and moderate
$n$. We use naive prefix-sum based data delivery without randomization since we
currently only use random inputs anyway -- for these the naive algorithm
coincides with the deterministic algorithm since all pieces are large with high
probability.

AMS-sort implements overpartitioning, however using the simple sequential
algorithm for bucket grouping which incurs an avoidable factor $\Oh{\log
  n}$. Also, information stemming from overpartitioning is not yet exploited for
the recursive subproblems.  This means that overpartitioning is not yet as
effective as it would be in a full-fledged implementation.

We divide each level of the algorithms into four distinct phases: splitter
selection, bucket processing (multiway merging or distribution), data delivery, and local sorting. To measure the
time of each phase, we place a MPI barrier before each phase. Timings for these
phases are accumulated over all recursion levels.

The time for building MPI communicators (which can be considerable) is not
included in the running time since this can be viewed as a precomputation that
can be reused over arbitrary inputs.

The algorithms are written in \verb!C++11! and compiled with version~$15.0$ of
the Intel \verb!icpc! compiler, using the full optimization flag \emph{-O3} and
the instruction set specified with \emph{-march=corei7-avx}. For inter-process
communication, we use version version~$1.3$ of the IBM \verb!mpich2! library.

During the bucket processing phase of RLM-sort, we use the
\verb!sequential_multiway_merge! implementation of the GNU Standard C++ Library
to merge buckets \cite{SSP07}. We used our own implementation of multisplitter partitioning
in the bucket processing phase, borrowed from super scalar sample
sort~\cite{SW04}. 

For the data delivery phase, we use our own implementation of a 1-factor
algorithm~\cite{SanTra02www} and compare it against the all-to-allv
implementation of the IBM \verb!mpich2! library. The 1-factor implementation
performs up to $p$ pairwise \texttt{MPI\_Isend} and \sloppy{\texttt{MPI\_Irecv}}
operations to distribute the buckets to their target groups. In contrast to the
\verb!mpich2! implementation, the 1-factor algorithm omits the exchange of empty
messages. We found that the 1-factor implementation is more stable and exchanges
data with a higher throughput on the average. Local sorting uses
\verb!std::sort!.

\subsection{Weak Scaling Analysis}

\begin{table}[t]
\begin{center}
\begin{tabular*}{\columnwidth}{@{\extracolsep{\fill} }c|rrrrr}
                              & level & \multicolumn{4}{c}{$p$}   \\
$k$                                 & & $512$ & $2048$ & $8192$ & $32768$ \\
\hline
\parbox[c]{0.7cm}{\centering$1$} & 1 & $16$  & $16$   & $16$   & $16$    \\\hline
\parbox[c]{0.7cm}{\centering$2$} & 1 & $32$  & $128$  & $512$  & $2048$  \\
                                 & 2 & $16$  & $16$   & $16$   & $16$    \\\hline
\parbox[c]{0.7cm}{\centering$3$} & 1 & $8$   & $16$   & $32$   & $64$    \\
                                 & 2 & $4$   & $8$    & $16$   & $32$    \\
                                 & 3 & $16$  & $16$   & $16$   & $16$    \\
\end{tabular*}
\end{center}
\caption{Selection of $r$ for weak scaling experiments}\label{tab:level_selection}
\end{table}

\begin{table}[t]
\begin{center}
\begin{tabular*}{\columnwidth}{@{\extracolsep{\fill} }c|rrrr}
 & \multicolumn{4}{c}{$p$} \\
$n/p$& $512$ & $2048$ & $8192$ & $32768$ \\
\hline
$10^5$      & $0.0228$ & $0.0277$ & $0.0359$ & $0.0707$ \\
$10^6$      & $0.2212$ & $0.2589$ & $0.2687$ & $0.9171$ \\
$10^7$      & $2.6523$ & $2.9797$ & $4.0625$ & $6.0932$ \\
\end{tabular*}
\end{center}
\caption{AMS-sort median wall-times of weak scaling experiments in seconds}\label{tab:wall_time}
\end{table}

\begin{figure}[t]
  \begin{center}
    \includegraphics[width=\columnwidth]{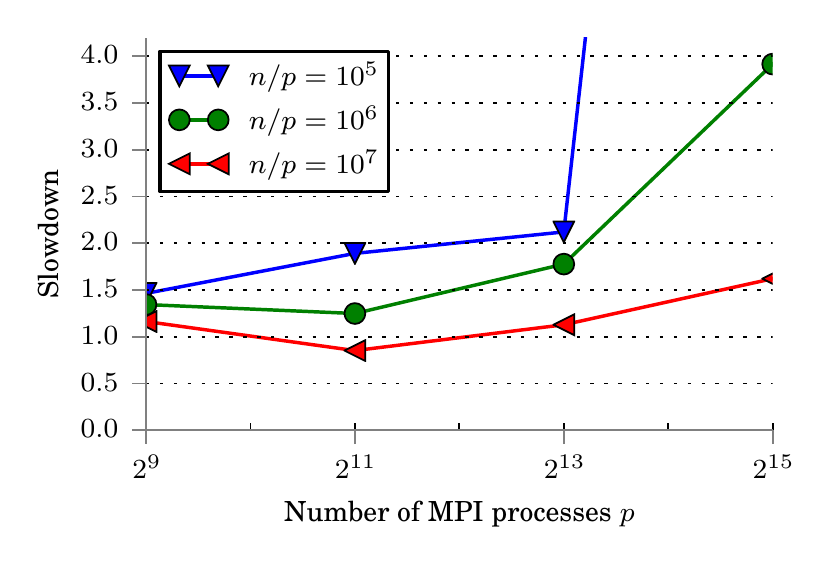}
    \caption{Slowdown of RLM-sort compared to AMS-sort based on optimal level choice}\label{fig:rlm_slowdown}
  \end{center}
\end{figure}

\begin{figure}[h!]
\subfloat{\includegraphics[width=\columnwidth]{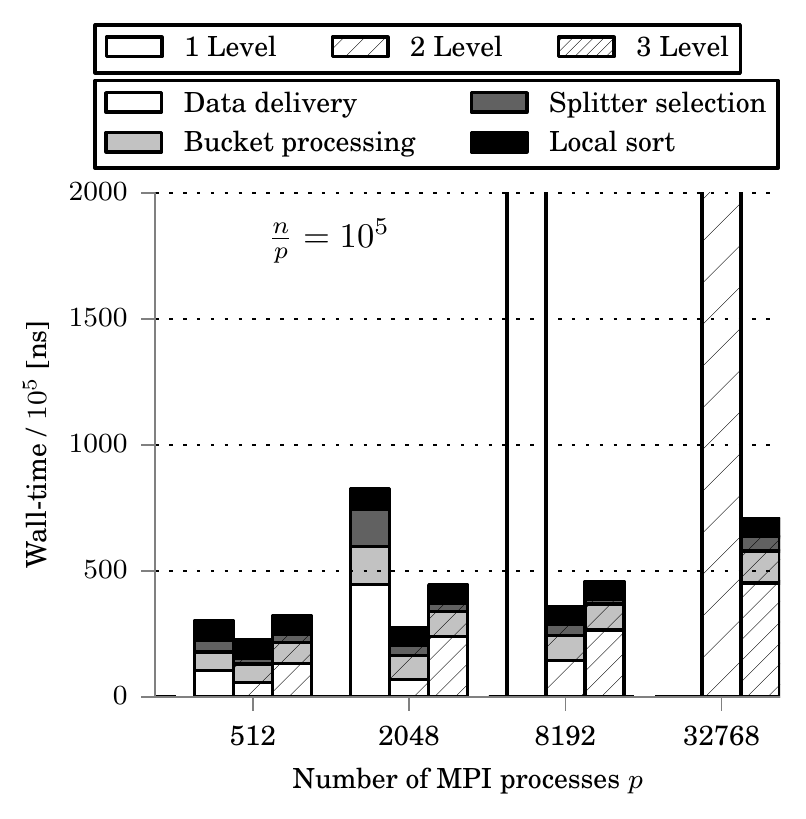}}\\[-10mm]
\subfloat{\includegraphics[width=\columnwidth]{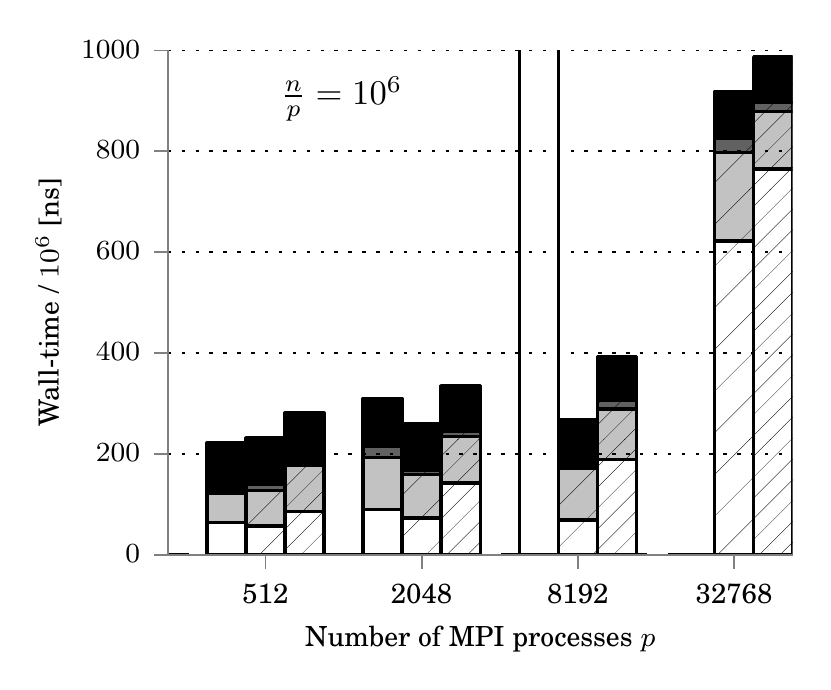}}\\[-12mm]
\subfloat{\includegraphics[width=\columnwidth]{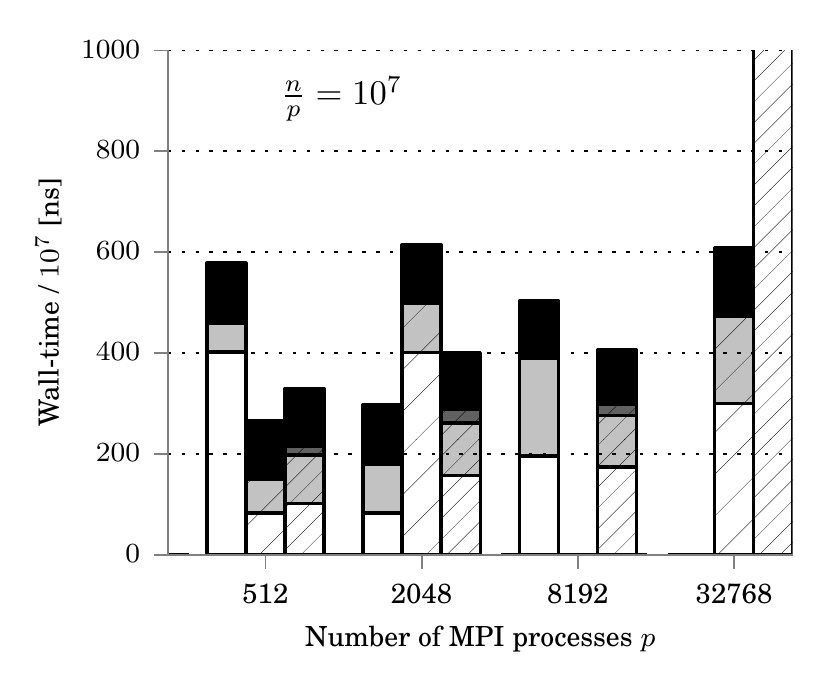}}\\
\caption{Weak scaling with $10^5$, $10^6$, and $10^7$ elements per MPI process of AMS-sort}\label{fig:runningtime_bars}
\end{figure}

The experimental setting of the weak scaling test is as follows: We benchmarked
AMS-sort at $32$, $128$, $512$, and $2048$ nodes. Each node executed $16$ MPI
processes. This results in $512$, $2048$, $8192$, and $32768$ MPI processes. The
benchmark configuration for $2048$ nodes has been executed on four exclusively
allocated islands. 
Table~\ref{tab:level_selection} shows the level
configurations of our algorithm. AMS-sort, configured with more than one level, splits the
remaining processes into groups with a size of $16$ MPI processes at the second
to last level. Thereby, the last level communicates just node-internally. For the
$3$-level AMS-sort, we split the MPI processes at the first level into
$2^{\ceil{\log(p)/2}}$ groups. AMS-sort configured the splitter selection
phase with an overpartitioning factor of $b = 16$ and an
oversampling factor of $a = 1.6 \: \log_{10}n$.

Figure~\ref{fig:runningtime_bars} details the wall-time of AMS-sort up to three
levels. For each wall-time, we show the proportion of time taken by each phase. The depicted wall-time is the
median of five measurements. 
Figure~\ref{fig:runningtime_distribution} in the Appendix shows the distribution of the wall-times.
Observe that AMS-sort is not limited by the
splitter selection phase in all test cases. In most cases, AMS-sort with more
than one level decreases the wall-time up to $8192$ MPI processes. Also, there
is a speedup in the data delivery phase and no significant slowdown in the
bucket processing phase due to cache effects. In these cases, the cost for
partitioning the data and distributing more than once is compensated by the decreased number of startups.
For the smaller volume of $10^5$ elements per MPI process, note that $3$-level AMS-sort 
is much faster than $2$-level AMS-sort in our experimental setup; the effect is reversed for more elements.
Note that there is inter-island data delivery at the first and second
level of $3$-level AMS-sort. The slowdown of sorting $10^6$ elements per MPI
process with $3$-level AMS-sort compared to $2$-level AMS-sort is small. So we
assume that the three level version becomes faster than the two level version
executed at more than four islands. In that case, it is more reasonable to set the
number of groups in the first level equal the amount of islands. This
results in inter-island communication just within the first level.

Table~\ref{tab:wall_time} depicts the median wall-time of our weak scaling
experiments of AMS-sort. Each entry is selected based on the level which
performed best. For a fixed $p$, the wall-time increases almost linear with the
amount of elements per MPI process. One exception is the wall-time for $8192$
nodes and $10^7$ elements. We were not able to measure the $2$-level AMS-sort as
the MPI-implementation failed  during this experiment. The wall-time
increases by a small factor up to $8192$ MPI processes for increasing
$p$. Executed with $32768$ MPI processes, AMS-sort is up to $3.5$ times slower
compared to intra-island sorting, allocated at one whole island. The slowdown
can be feasibly explained by the interconnect which connects islands among each
other. The interconnect has an bandwidth ratio of $4 : 1$ compared to the
intra-island interconnect.

Generally, for large $p$, the execution time fluctuates a lot (also see
Figure~\ref{fig:runningtime_distribution}).  This fluctuation is almost
exclusively within the all-to-all exchange. Further research has to show to what
extent this is due to interference due to network traffic of other applications
or suboptimal implementation of all-to-all data exchange. Both effects seem to
be independent of the sorting algorithm however.

Figure~\ref{fig:rlm_slowdown} illustrates the slowdown of RLM-sort compared to
AMS-sort. For each algorithm, we selected the number of levels with the best
wall-time.  Note that the slowdown of
RLM-sort is higher than one in almost all test cases. The slowdown is
significantly increased for small $n$ and large $p$. This observation matches with
the isoefficiency function of RLM-sort which is a $\log^2 p$ factor worse than
the isoefficiency function of AMS-sort.

\FloatBarrier

\subsection{Comparison with Other Implementations}

Comparisons to other systems are difficult, since it is not easy to simply
download other people's software and to get it to run on a large machine.
Hence, we have to compare to the literature and the web.  Our absolute running
times for $n=10^7p$ are similar to those of observed in Solomonik and Kale
\cite{SolKal10} for $n=8\cdot 10^6\cdot p$ on a CrayXT 4 with up to $2^{15}$
PEs. This machine has somewhat slower PEs (2.1 GHz AMD Barcelona) but higher
communication bandwidth per PE.  No running times for smaller inputs are
given. It is likely that there the advantage of multilevel algorithms such as
ours becomes more visible. Still, we view it as likely that adapting their
techniques for overlapping communication and sorting might be useful for our
codes too.

A more recent experiment running on even more PEs is MP-sort  \cite{FSSMC14}.
MP-sort is a single-level multiway mergesort that implements local multiway
merging by sorting from scratch. They run the same weak scaling test as us using
up to 160\,000 cores of a Cray XE6 (16 AMD Opteron cores $\times$ 2 processors
$\times$ 5\,000 nodes).  This code is much slower than ours. For 
$n=10^5\cdot p$, and $p=2^{14}$ the code needs 20.45 seconds -- 289 times more than
ours for $p=2^{15}$. When going to $p=80\,000$ the
running time of MP-sort goes up by another order of magnitude. At large $p$,
MP-sort is hardly slower for larger inputs (however still about six times slower
than AMS-sort). This is a clear indication that a single level algorithm does
not scale for small inputs.

Different but also interesting is the Sort Benchmark which is quite established
in the data base community (\url{sortbenchmark.org}). The closest category is
Minute-Sort. The 2014 winner, Baidu-Sort (which uses the same algorithm as
TritonSort \cite{RasEtAl11}), sorts 7 TB of data (100 byte elements with 10 byte
random keys) in 56.7s using 993 nodes with two 8-core processors (Intel Xeon
E5-2450, 2.2 GHz) each ($p=15\,888$).  Compared to our experiment at
$n=10^7\cdot 2^{15}$, they use about half as many cores as us, and sort about
2.7 times more data.  On the other hand, Baidu-Sort takes about 9.3
times longer than our 2-level algorithm. and we sort
about 5 times more (8-byte) elements. Even disregarding that we also sort about
5 times more (8-byte) elements, this leaves us being about two times more
efficient. This comparison is unfair to some extent since Minute-Sort requires
the input to be read from disk and the result to be written to disk.
However, the machine used by Baidu-Sort has 993$\times$8 hard disks. At a typical
transfer rate of 150 MB/s this means that, in principle, it is possible to read
and write more than 30 TB of data within the execution time. Hence, it seems
that also for Baidu-Sort, the network was the major performance bottleneck.

\section{Conclusion}

We have shown how practical parallel sorting algorithms like multi-way mergesort
and sample sort can be generalized so that they scale on massively parallel
machines without incurring a large additional amount of communication volume.
Already our prototypical implementation of AMS-sort shows very competitive
performance that is probably the best by orders of magnitude for large $p$ and
moderate $n$. For large $n$ it can compete with the best single-level algorithms.

Future work should include experiments on more PEs, a native shared-memory
implementation of the node-local level, a full implementation of data delivery,
faster implementation of overpartitioning, and, at least for large $n$, more
overlapping of communication and computation. However, the major open problem
seems to be better data exchange algorithms, possibly independently of the
sorting algorithm.

\label{s:conclusion}
\paragraph*{Acknowledgments:}
The authors gratefully acknowledge the Gauss Centre for Supercomputing e.V. (\href{www.gauss-centre.eu}{www.gauss-centre.eu}) for funding this project by providing computing time on the GCS Supercomputer SuperMUC at Leibniz Supercomputing Centre (LRZ, \href{www.lrz.de}{www.lrz.de})
Special thanks go to SAP AG, Ingo Mueller, and Sebastian Schlag for making their 1-factor algorithm~\cite{SSM13} available.
Additionally, we would like to thank Christian Siebert for valuable discussions.

\bibliographystyle{abbrv}
\bibliography{diss}

\begin{thebibliography}{10}

\bibitem{ABSS14}
M.~{Axtmann}, T.~{Bingmann}, P.~{Sanders}, and C.~{Schulz}.
\newblock {Practical Massively Parallel Sorting -- Basic Algorithmic Ideas}.
\newblock {\em Preprint arXiv:1410.6754v1}, Oct. 2014.

\bibitem{BalEtAl95}
V.~Bala, J.~Bruck, R.~Cypher, P.~Elustondo, A.~Ho, C.~Ho, S.~Kipnis, and
  M.~Snir.
\newblock {CCL}: A portable and tunable collective communication library for
  scalable parallel computers.
\newblock {\em IEEE Transactions on Parallel and Distributed Systems},
  6(2):154--164, 1995.

\bibitem{Batcher68}
K.~E. Batcher.
\newblock Sorting networks and their applications.
\newblock In {\em AFIPS Spring Joint Computing Conference}, pages 307--314,
  1968.

\bibitem{BDH95}
A.~B{\"a}umker, W.~Dittrich, and F.~Meyer auf~der Heide.
\newblock Truly efficient parallel algorithms: $c$-optimal multisearch for an
  extension of the {BSP} model.
\newblock In {\em Algorithms — ESA'95}, pages 17--30. Springer, 1995.

\bibitem{BES14}
T.~Bingmann, A.~Eberle, and P.~Sanders.
\newblock Engineering parallel string sorting.
\newblock {\em Preprint arXiv:1403.2056}, 2014.

\bibitem{BleEtAl91short}
G.~E. Blelloch et~al.
\newblock A comparison of sorting algorithms for the connection machine {CM-2}.
\newblock In {\em 3rd Symposium on Parallel Algorithms and Architectures},
  pages 3--16, 1991.

\bibitem{Borkar13}
S.~Borkar.
\newblock Exascale computing -- a fact or a fiction?
\newblock Keynote presentation at IPDPS 2013, Boston, May 2013.

\bibitem{BFV04}
G.~S. Brodal, R.~Fagerberg, and K.~Vinther.
\newblock Engineering a cache-oblivious sorting algorithm.
\newblock In {\em 6th Workshop on Algorithm Engineering and Experiments}, 2004.

\bibitem{Col88}
R.~Cole.
\newblock Parallel merge sort.
\newblock {\em SIAM Journal on Computing}, 17(4):770--785, 1988.

\bibitem{DSSS04}
R.~Dementiev, P.~Sanders, D.~Schultes, and J.~Sibeyn.
\newblock Engineering an external memory minimum spanning tree algorithm.
\newblock In {\em IFIP TCS}, pages 195--208, Toulouse, 2004.

\bibitem{DubPriRan96}
D.~Dubhashi, V.~Priebe, and D.~Ranjan.
\newblock Negative dependence through the {FKG} inequality.
\newblock Research Report MPI-I-96-1-020, Max-Planck-Institut f{\"u}r
  Informatik, Im Stadtwald, D-66123 Saarbr{\"u}cken, Germany, Aug. 1996.

\bibitem{FSSMC14}
Y.~Feng, M.~Straka, T.~di~Matteo, and R.~Croft.
\newblock {MP}-sort: Sorting at scale on blue waters.
\newblock {\url{https://www.writelatex.com/read/sttmdgqthvyv}} accessed Jan 17,
  2015, 2014.

\bibitem{GerVal94}
A.~Gerbessiotis and L.~Valiant.
\newblock Direct bulk-synchronous parallel algorithms.
\newblock {\em Journal of Parallel and Distributed Computing}, 22(2):251--267,
  1994.

\bibitem{Goo99}
M.~T. Goodrich.
\newblock Communication-efficient parallel sorting.
\newblock {\em SIAM Journal on Computing}, 29(2):416--432, 1999.

\bibitem{HagRue89}
T.~Hagerup and C.~R{\"u}b.
\newblock Optimal merging and sorting on the {EREW-PRAM}.
\newblock {\em Information Processing Letters}, 33:181--185, 1989.

\bibitem{Hoa61b}
C.~A.~R. Hoare.
\newblock Algorithm 65 (find).
\newblock {\em Communication of the ACM}, 4(7):321--322, 1961.

\bibitem{HMS15}
L.~H{\"u}bschle-Schneider, I.~M{\"u}ller, and P.~Sanders.
\newblock Communication efficient algorithms for top-k selection problems.
\newblock submitted for SPAA 2015, 2015.

\bibitem{IKS09}
M.~Ikkert, T.~Kieritz, and P.~Sanders.
\newblock Parallele {A}lgorithmen.
\newblock course notes, October 2009.

\bibitem{Jaj92}
J.~J{\'a}j{\'a}.
\newblock {\em An Introduction to Parallel Algorithms}.
\newblock Addison Wesley, 1992.

\bibitem{Knu98short}
D.~E. Knuth.
\newblock {\em The Art of Computer Programming---Sorting and Searching}.
\newblock Addison Wesley, 1998.

\bibitem{KumEtAl94}
V.~Kumar, A.~Grama, A.~Gupta, and G.~Karypis.
\newblock {\em Introduction to Parallel Computing. Design and Analysis of
  Algorithms}.
\newblock Benjamin/Cummings, 1994.

\bibitem{LiSev94}
H.~Li and K.~C. Sevcik.
\newblock Parallel sorting by overpartitioning.
\newblock In {\em ACM Symposium on Parallel Architectures and Algorithms},
  pages 46--56, Cape May, New Jersey, 1994.

\bibitem{LubRac88}
M.~Luby and C.~Rackoff.
\newblock How to construct pseudorandom permutations from pseudorandom
  functions.
\newblock {\em SIAM Journal on Computing}, 17(2):373--386, Apr. 1988.

\bibitem{MehSan08}
K.~Mehlhorn and P.~Sanders.
\newblock {\em Algorithms and Data Structures --- The Basic Toolbox}.
\newblock Springer, 2008.

\bibitem{NaoRei99}
M.~Naor and O.~Reingold.
\newblock On the construction of pseudorandom permutations: {Luby-Rackoff}
  revisited.
\newblock {\em Journal of Cryptology: the journal of the International
  Association for Cryptologic Research}, 12(1):29--66, 1999.

\bibitem{RasEtAl11}
A.~Rasmussen, G.~Porter, M.~Conley, H.~V. Madhyastha, R.~N. Mysore, A.~Pucher,
  and A.~Vahdat.
\newblock Tritonsort: A balanced large-scale sorting system.
\newblock In {\em NSDI}, 2011.

\bibitem{San00b}
P.~Sanders.
\newblock Fast priority queues for cached memory.
\newblock {\em ACM Journal of Experimental Algorithmics}, 5, 2000.

\bibitem{San08VLPA}
P.~Sanders.
\newblock Course on {Parallel Algorithms}, lecture notes, 2008.
\newblock \url{http://algo2.iti.kit.edu/sanders/courses/paralg08/}.

\bibitem{SSM13}
P.~Sanders, S.~Schlag, and I.~M{\"u}ller.
\newblock Communication efficient algorithms for fundamental big data problems.
\newblock In {\em IEEE Int. Conf. on Big Data}, 2013.

\bibitem{SST09}
P.~Sanders, J.~Speck, and J.~L. Tr{\"{a}}ff.
\newblock Two-tree algorithms for full bandwidth broadcast, reduction and scan.
\newblock {\em Parallel Computing}, 35(12):581--594, 2009.

\bibitem{SanTra02www}
P.~Sanders and J.~L. Tr{\"a}ff.
\newblock The factor algorithm for regular all-to-all communication on clusters
  of {SMP} nodes.
\newblock In {\em 8th Euro-Par}, number 2400, pages 799--803. Springer
  {\copyright}, 2002.

\bibitem{SW04}
P.~Sanders and S.~Winkel.
\newblock Super scalar sample sort.
\newblock In {\em 12th European Symposium on Algorithms}, volume 3221 of {\em
  LNCS}, pages 784--796. Springer, 2004.

\bibitem{SSP07}
J.~Singler, P.~Sanders, and F.~Putze.
\newblock {MCSTL}: The multi-core standard template library.
\newblock In {\em 13th Euro-Par}, volume 4641 of {\em LNCS}, pages 682--694.
  Springer, 2007.

\bibitem{SolKal10}
E.~Solomonik and L.~Kale.
\newblock Highly scalable parallel sorting.
\newblock In {\em IEEE International Symposium on Parallel Distributed
  Processing (IPDPS)}, pages 1--12, April 2010.

\bibitem{Val94}
L.~Valiant.
\newblock A bridging model for parallel computation.
\newblock {\em Communications of the ACM}, 33(8):103--111, 1994.

\bibitem{VSIR91}
P.~J. Varman et~al.
\newblock Merging multiple lists on hierarchical-memory multiprocessors.
\newblock {\em J. Par. {\&} Distr. Comp.}, 12(2):171--177, 1991.

\end{thebibliography}

\clearpage
\appendix
\section{Randomized Data Delivery}\label{app:randomDelivery}
\begin{figure*}\normalsize\centering
  \input{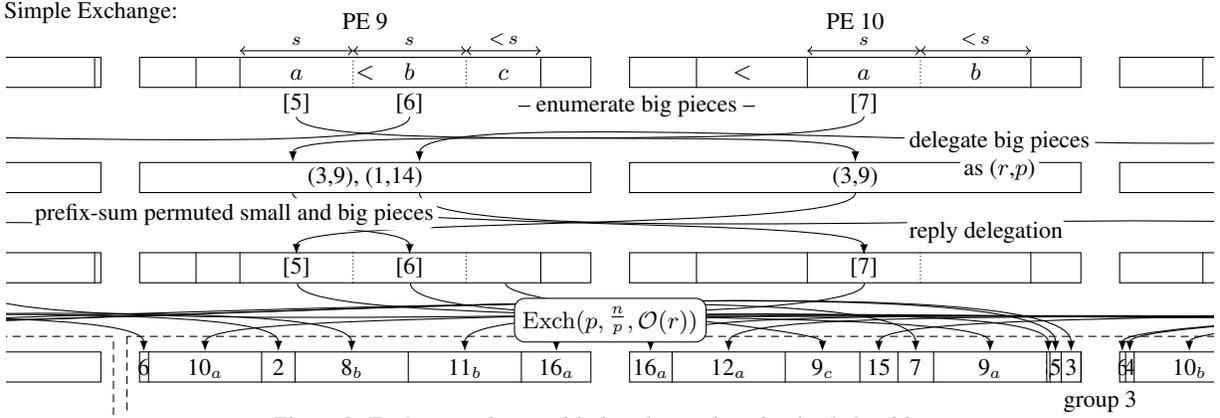}
  \caption{Exchange schema with the advanced randomized algorithm.}\label{fig:exchange schema2}
\end{figure*}

Our advanced randomized data delivery 
algorithm is asymptotically more efficient than the simple one described in Section~\ref{ss:deliver}
and it may be simpler to implement than the deterministic one from Section~\ref{sss:deterministic}
since no parallel merging operation is needed.
Compared to the simple algorithm, the algorithm adds more randomization and invests some additional
communication.  The idea is to break large pieces into several smaller pieces. A
piece whose size $x$ exceeds a limit $s$ is broken into $\floor{x/s}$ pieces of
size $s$ and one piece of size $x\bmod s$. We set $s := an/rp$ to be $a$ times the average piece size $n/rp$ where $a$ is a tuning parameter to be
chosen later.  The resulting small pieces (size below $s$) stay where they are
and the random permutation of the PE numbers takes care of their random
placement. The large pieces are delegated to another (random) PE using a further
random permutation. This is achieved by enumerating them globally over all parts
using a prefix sum.  Suppose there are $K$ large pieces, then we use a
pseudorandom permutation $\pi:0..K-1\rightarrow 0..K-1$ to delegate piece $i$ to
PE $1+\pi(i)\bmod p$. Note that this assignment only entails to tell PE $j$
about the origin of this piece and its target group -- there is no need to move
the actual elements at this point. In Figure~\ref{fig:exchange schema2}, we denote
the delegation tuples with origin PE $p$ and target group $r$ as $(r,p)$.
Next, for each part, a PE reorders its
small pieces and delegated large pieces randomly (of course without choosing the intra-piece sorting). Only then, a prefix sum is
used to enumerate the elements in each part. The ranges of numbers assigned to
the pieces are then communicated back to the PEs actually holding the data and
we continue as in the basic approach -- computing target PEs based on the
received ranges of numbers.

\begin{lemma}\label{lem:delegate}
The two stage approach needs time
$\Oh{\Tstart\log p+r\Tword}+2\Exch(p,\Oh{r/a},\ceil{r/a})$
to assign data to target PEs.
\end{lemma}
\begin{proof}
  Each PE will produce at most $\frac{n/p}{s}=\frac{n/p}{an/rp}=r/a$ large
  pieces. Overall, there will be at most $\frac{n}{s}=\frac{n}{an/rp}=pr/a$
  large pieces. The random mapping will delegate at most $\ceil{r/a}$ of these
  messages to each PE with high probability. Since each delegation and notification message has
  constant size, $2\Exch(p,\Oh{r/a},\ceil{r/a})$ accounts for the resulting
communication costs.  All involved prefix sums are vector valued prefix sum with
vector length $r$ and can thus be implemented to run in time $\Oh{\Tstart \log
  p+r\Tword}$. This term also covers the local computations.
\end{proof}

\begin{lemma}
  No PE sends more than $2r(1+1/a)$ messages during the main data exchange of
  one phase of RLM-sort.  Moreover, the total number of messages for a single
  part is at most $p(1+1/r+1/a)$.
\end{lemma}
\begin{proof}
  As shown above, each PE produces at most $r(1+1/a)$ pieces, each of which may
  be split into at most two messages.  For each part, there are at most $p$ small
  pieces and $\frac{n/r}{an/rp}=p/a$ large pieces.  At most $p/r-1< p/r$ pieces
  can be split because their assigned range of element numbers intersects the
  ranges of responsibility of two PEs. Overall, we get $p(1+1/r+1/a)$ messages
  per part.
\end{proof}

\begin{lemma}\label{lem:receive}
  Assuming that our pseudorandom permutations behave like truly random
  permutations, with probability $1-\Oh{1/p}$, no PE receives more than
  $1+2r(1+1/a)$ messages during one phase of RLM-sort for some value of
  $a\in\Th{\sqrt{r/\log p}}$.
\end{lemma}
\begin{proof}
  Let $m\leq p(1+1/a)$ denote the number of pieces generated for part $x$.  It
  suffices to prove that the probability that any of the PEs responsible for it
  receives more than $1+2mr/p\leq 2r(1+1/a)$ messages is at most $1/rp$ for an
  appropriate constant.  We now abstract from the actual implementation of data
  assignment by observing that the net effect of our randomization is to produce
  a random permutation of the pieces involved in each part.\TODO{do we need
    to prove this?}  In this abstraction, the ``bad'' event can only occur if
  the permutation produces $2mr/p$ consecutive pieces of total size at most
  $n/p$.  More formally, let $X_1$,\ldots,$X_m$ denote the piece sizes. The
  $X_i$ are random variables with range $[0,\frac{an}{rp}]$ and
  $\sum_iX_i=n/r$. The randomness stems from the random permutation determining
  the ordering. Unfortunately, the $X_i$ are not independent of each
  other. However, they are negatively associated \cite{DubPriRan96}, i.e., if
  one variable is large, then a different variable tends to be smaller.  In this
  situation, Chernoff-Hoeffding bounds for the probability that a sum deviates
  from its expectation still apply. Now, for a fixed $j$, consider
  $X\Is\sum_{j\leq i < j+2mr/p}X_i$. It suffices to show that $\prob{X<n/p}\leq
  1/rpm$ -- in that case, the probability that the bad event occurs for some $j$
  is at most $1/rp$. We have $\expect[X]=2n/p$ which differs by $t\Is n/p$ from
  the bound marking a bad event.  Hoeffding's inequality then assures that the
  probability of the bad event is at most
  \begin{align*}
  \prob{X<n/p}&\leq 2e^{-\frac{2t^2}{\frac{2mr}{r}\cdot\left(\frac{an}{rp}\right)^2}}
                 =   2e^{-\frac{pr}{ma^2}}
                \leq 2e^{-\frac{pr}{a^2+1}}\punkt
  \end{align*}
  This should be smaller than $1/rp$.
  Solving the resulting relation for $a$ yields
  $$a\leq \frac{1}{2}\left(\sqrt{1+\frac{r}{\ln\frac{rp}{2}}}-1\right)\punkt$$
\end{proof}
Note that Lemma~\ref{lem:receive} implies that with high probability both the
number of sent and received messages during data exchange will be close to $2r$
and the number of message startups for delegating pieces (see
Lemma~\ref{lem:delegate}) will be $o(r)$.  Hence, we have shown that handling
worst case inputs by our algorithm adds only lower order cost terms compared to
the simple variant (plain prefix sums without any randomization) on average case
inputs. In contrast, applying the simple approach to worst case inputs directly,
completely ruins performance.

We summarize the result in the following theorem:
\begin{theorem}
Data delivery of $r\times p$ pieces to $r$ parts can be implemented to run in time
$$\ExchTilde(p,\tfrac{n}{p},2r)$$
with high probability.
\end{theorem}

\section{Pseudorandom Permutations}\label{app:randperm}
During redistribution of data, we will randomize the rearrangement to avoid bad
cases. For this, we select a pseudo-random permutation, which can be constructed,
e.g., by composing three to four Feistel permutations~\cite{LubRac88,DSSS04}. We
adapt the description from \cite{DSSS04} to our purposes.

Assume we want to compute a permutation
$\pi:0..n-1\rightarrow 0..n-1$. Assume for now that $n$ is a square so
that we can represent a number $i$ as a pair $(a,b)$ with
$i=a+b\sqrt{n}$.  Our permutations are constructed from \emph{Feistel}
permutations, i.e., permutations of the form
$\pi_f((a,b))=(b,a+f(b)\bmod \sqrt{n})$ for some pseudorandom mapping
$f:0..\sqrt{n}-1\rightarrow 0..\sqrt{n}-1$.
$f$ can be any hash function that behaves reasonably similar to a random function in practice.
 It is known that a permutation
$\pi(x)=\pi_f(\pi_g(\pi_h(\pi_l(x))))$ build
by chaining four Feistel permutations is ``pseudorandom'' in a
sense useful for cryptography. The same holds if the innermost and outermost
permutation is replaced by an even simpler permutation \cite{NaoRei99}.
In \cite{DSSS04}, we used just two stages of Feistel-Permutations.

A permutation $\pi'$ on $0..\ceil{\sqrt{n}}^2-1$ can be transformed
to a permutation $\pi$ on $0..n-1$ by
iteratively applying $\pi'$ until a value below $n$ is obtained.
Since $\pi'$ is a permutation, this process must
eventually terminate. If $\pi'$ is random, the expected number
of iterations is close to $1$ and
it is unlikely that more than three iterations are necessary

Since the description of $\pi$ requires very little state, we can replicate this state over all PEs.

\section{Accelerating Bucket Grouping}\label{app:scan}
The first observation for improving the binary search algorithm from
Section~\ref{s:sample} is that a PE-group size can take only $\Oh{(br)^2}$
different values since it is defined by a range of buckets.  We can modify the
binary search in such a way that it operates not over all conceivable group
sizes but only over those corresponding to ranges of buckets.  When a scanning
step succeeds, we can safely reduce the upper bound for the binary search to the
largest PE-group actually used. On the other hand, when a scanning step fails,
we can increase the lower bound: during the scan, whenever we finish a PE-group
of size $x$ because the next bucket of size $y$ does not fit (i.e., $x+y>L$), we
compute $z=x+y$. The minimum over all observed $z$-values is the new lower
bound. This is safe, since a value of the scanning bound $L$ less then $z$ will
reproduce the same failed partition. 
This already yields an algorithm running in time $\Oh{br\log(br)^2}=\Oh{br \log(br)}$.

The second observation is only values for $L$ in the range
$\ceil{n/r-1}..(1+\Oh{1/b})n/r$ are relevant (see
Lemma~\ref{lm:ams-param}). Only $\Oh{br}$ bucket ranges will have a total size
in this range. To see this, consider any particular starting bucket for a bucket
range. Searching  from there to the right for range end points, we can skip all end buckets where the
total size is below $n/r$. We can stop as soon as the total size leaves the relevant range.  
Since buckets have average size $\Oh{n/b}$, only a
constant number of end points will be in the relevant range on the average.
Overall, we get $\Oh{br}\cdot\Oh{1}=\Oh{br}$ relevant bucket ranges.
Using this for initializing the binary search, saves a factor about two for the
sequential algorithm.

Using all $p$ available PEs, we can do even better:
in each iteration, we split the remaining range for $L$ evenly into $p+1$
subranges.  Each PE tries one subrange end point for scanning and uses the first observation
to round up or down to an actually occurring size of a bucket range. Using a
reduction we find the largest $L$-value $\Lmin$ for a failed scan and the
smallest $L$ value $\Lmax$ for a successful scan. When $\Lmax=\Lmin$ we have
found the optimal value for $L$. Otherwise, we continue with the range
$\Lmax..\Lmin$. Since the bucket range sizes in the feasible region are fairly
uniformly distributed, the number of iterations will be $\log_{p+1}\Oh{br}$.
Since $p\geq r$, this is $\Oh{1}$ if $b$ is polynomial in $r$. Indeed, one or
two iterations are likely to succeed in all reasonable cases.

\section{Tie Breaking for Key Comparisons}\label{app:tie}

Conceptually, we assign the key $(x,i,j)$ to an element with key $x$, stored on PE
$i$ at position $j$ of the input array. Using lexicographic ordering makes
the keys unique.  For a practical implementation, it is important not to do this
explicitly for every element.  We explain how this can be done for AMS-sort.
First note, that in AMS-sort there is no need to do tie breaking across levels
or for the final local sorting.  Sample sorting and splitter determination can
afford to do tie breaking explicitly, since these steps are more latency
bound. For partitioning, we can use a version of super scalar sample sort, that
also produces a bucket for elements equal to the splitter.  This takes only one
additional comparison \cite{BES14} per element. Only if an input element $x$
ends up in an equality bucket we need to perform the lexicographic
comparison. Note that at this point, the PE number for $x$ and its input
position are already present in registers anyway.

\section{Additional Experimental  Data}\label{app:furtherPlots}

\begin{figure}[t]
  \begin{center}
    \includegraphics[width=\columnwidth]{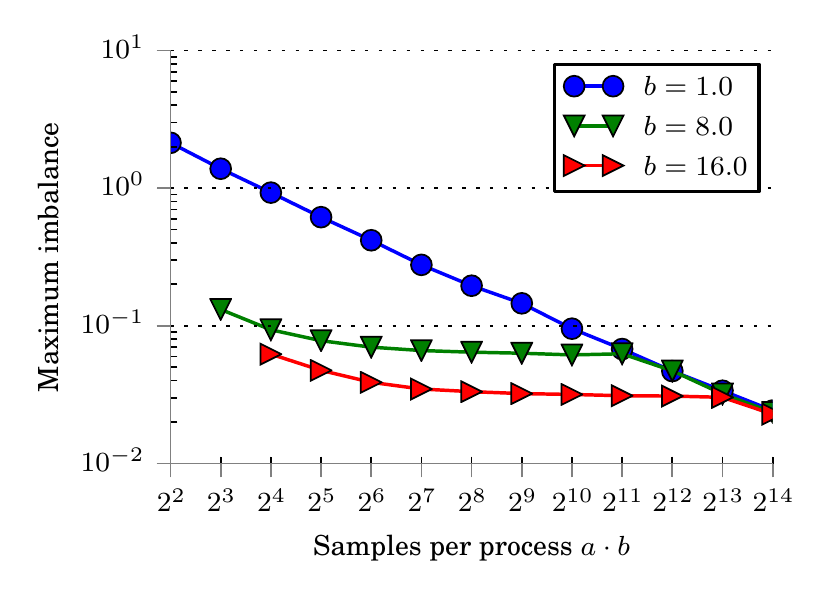}
    \caption{Maximum imbalance among groups of AMS-sort sorted sequences}\label{fig:max_imbalance}
  \end{center}
\end{figure}

\begin{figure}[t]
  \begin{center}
    \includegraphics[width=\columnwidth]{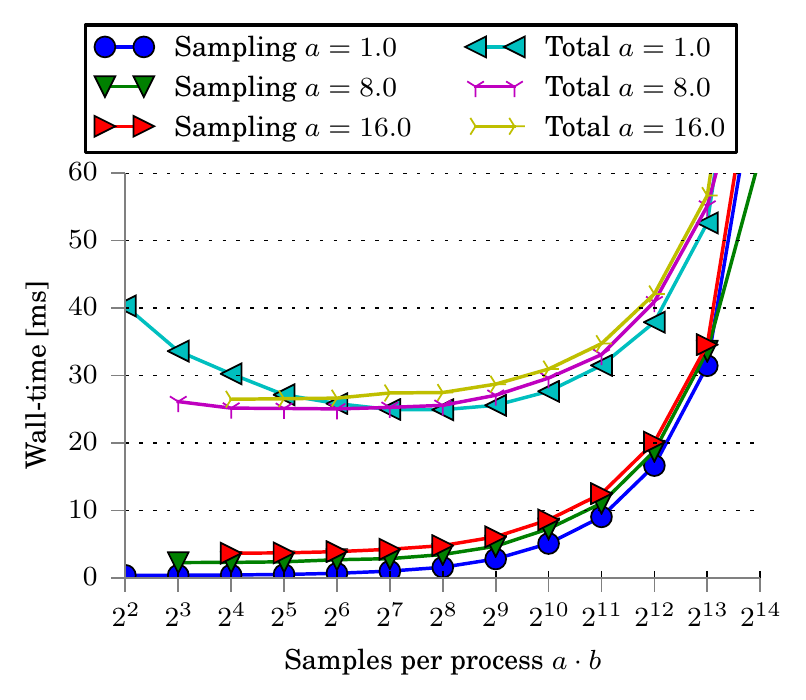}
    \caption{Wall-time of AMS-sort for various values of $a$ and $b$}\label{fig:influence_imbalance}
  \end{center}
\end{figure}

The overpartitioning factor $b$ influences the wall-time of AMS-sort. It has an effect on the
splitter selection phase itself but also an implicit impact on all other phases. 
To investigate this impact, we executed AMS-sort for various values of $b$ with $512$ MPI processes 
and $10^5$ elements each. Figure~\ref{fig:influence_imbalance} shows how the wall-time of AMS-sort depends on the number of samples per process $a \cdot b$. Depending on the oversampling factor $a$,
the wall-time firstly decreases as the maximum imbalance decreases. This leads to faster data delivery, 
bucket processing, and splitter selection phases. However, the wall-time increases for large $a$ as the
additional cost of the splitter selection phase dominates. On the one hand, AMS-sort performs best for 
an oversampling factor of $1$ and an overpartitioning factor of $64$. On the other hand, 
Figure~\ref{fig:max_imbalance} illustrates that the maximum imbalance is significantly higher for slightly
slower AMS-sort algorithms, configured with $b > 1$.

\begin{figure}[t]
\subfloat{\includegraphics[width=\columnwidth]{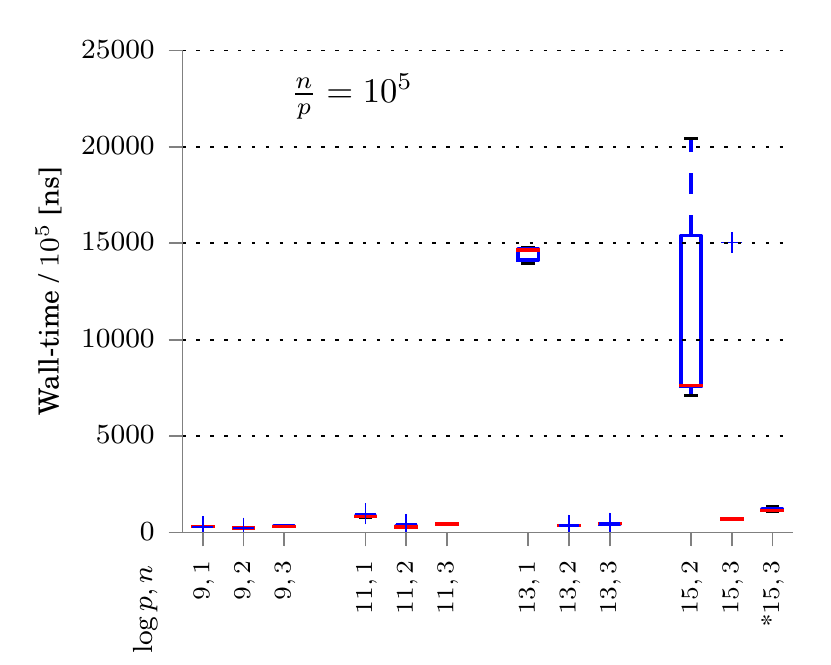}}\\
\subfloat{\includegraphics[width=\columnwidth]{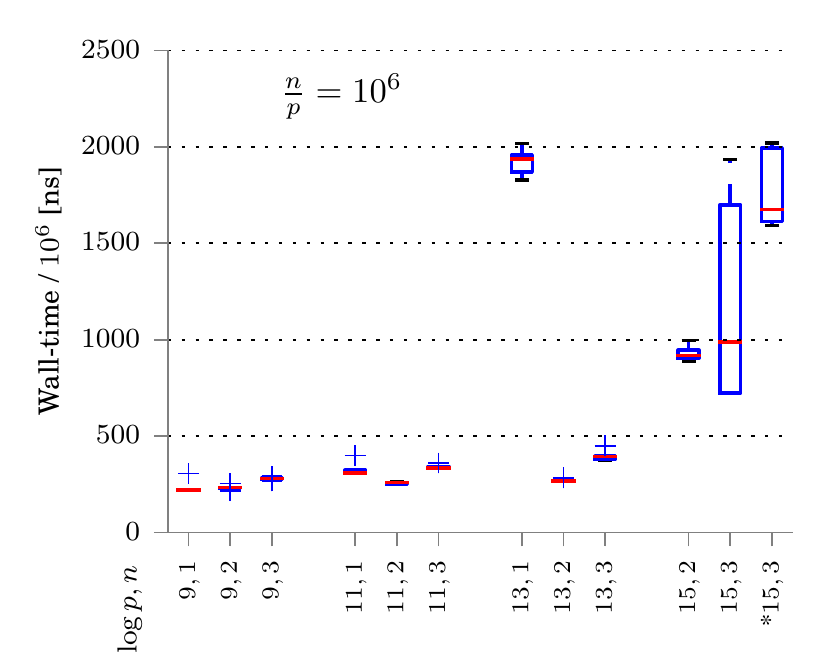}}\\
\subfloat{\includegraphics[width=\columnwidth]{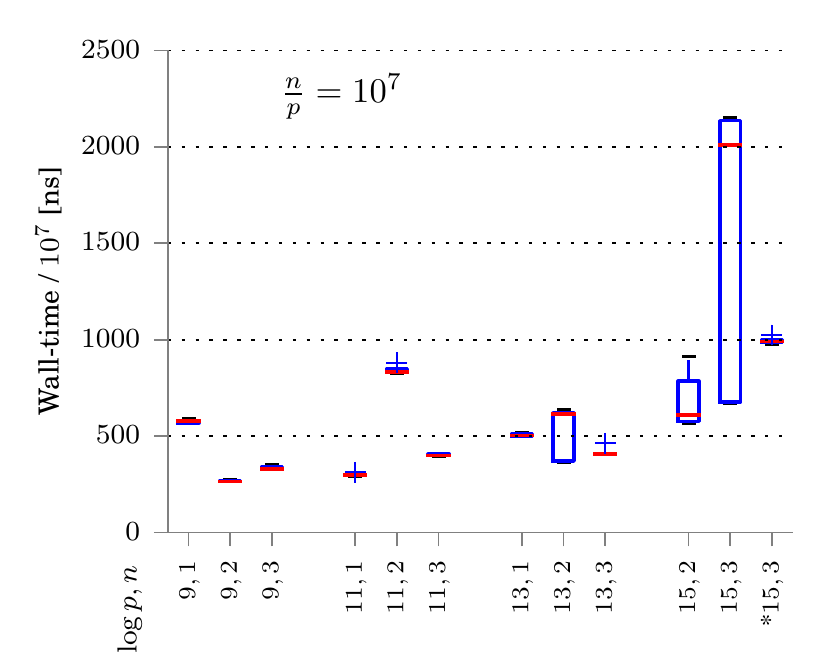}}\\
\caption{AMS-sort with $10^5$, $10^6$, and $10^7$ elements per MPI process}\label{fig:runningtime_distribution}
\end{figure}

\end{document}